\documentclass[journal,10pt,onecolumn]{IEEEtran}

\ifCLASSINFOpdf
\else
   \usepackage[dvips]{graphicx}
\fi
\usepackage{url}
\usepackage{algorithmicx}
\usepackage{algpseudocode}
\usepackage{algorithm}
\algdef{SE}[DOWHILE]{Do}{doWhile}{\algorithmicdo}[1]{\algorithmicwhile\ #1}%
\usepackage{amsmath}
\usepackage{dsfont, cuted}
\usepackage{amsfonts}
\usepackage{amsthm}
\usepackage[utf8]{inputenc}
\usepackage{breqn}
\usepackage{xcolor}
\usepackage{subfigure}
\usepackage{multirow}
\usepackage{multicol}
\usepackage{booktabs}
\usepackage{fullwidth}
\usepackage[noadjust]{cite}
\usepackage{setspace}
\usepackage{hyperref}
\usepackage{bmpsize}
\usepackage{graphicx}
\newtheorem{thm}{Theorem}
\newtheorem{lem}{Lemma}

\newtheorem{rem}{Remark}
\begin{document}

\title{Convergence Guarantees for Non-Convex Optimisation with Cauchy-Based Penalties}

\author{Oktay~Karakuş,~\IEEEmembership{Member,~IEEE,}
Perla~Mayo, Alin~Achim,~\IEEEmembership{Senior Member,~IEEE}
        \thanks{This work was supported in part by the UK Engineering and Physical Sciences Research Council (EPSRC) under grant EP/R009260/1 (AssenSAR), in part by a CONACyT PhD studentship under grant 461322 to Mayo, and in part by a Leverhulme Trust Research Fellowship to Achim (INFHER).}
        \thanks{Oktay Karakuş, Perla Mayo and Alin Achim are with the Visual Information Laboratory, University of Bristol, Bristol BS1 5DD, U.K. (e-mail: o.karakus@bristol.ac.uk; pm15334@bristol.ac.uk; alin.achim@britol.ac.uk)}
}

\maketitle

\begin{abstract}
In this paper, we propose a proximal splitting methodology with a non-convex penalty function based on the heavy-tailed Cauchy distribution. We first suggest a closed-form expression for calculating the proximal operator of the Cauchy prior, which then makes it applicable in generic proximal splitting algorithms. We further derive the condition required for guaranteed convergence to a solution in optimisation problems involving the Cauchy based penalty function. Setting the system parameters by satisfying the proposed condition ensures convergence even though the overall cost function is non-convex, when minimisation is performed via a proximal splitting algorithm. The proposed method based on Cauchy regularisation is evaluated by solving generic signal processing examples, i.e. 1D signal denoising in the frequency domain, two image reconstruction tasks including de-blurring and denoising, and error recovery in a multiple-antenna communication system. We experimentally verify the proposed convergence conditions for various cases, and show the effectiveness of the proposed Cauchy based non-convex penalty function over state-of-the-art penalty functions such as $L_1$ and total variation ($TV$) norms.
\end{abstract}

\begin{IEEEkeywords}
Non-convex regularisation; Convex optimisation; Cauchy proximal operator; Inverse problems; Image reconstruction; MIMO error recovery.
\end{IEEEkeywords}

\IEEEpeerreviewmaketitle

\section{Introduction}
\IEEEPARstart{T}{he} problem of estimating unknown physical properties directly from observations (e.g. measurements, data) arises in almost all signal/image processing applications. Problems of this kind are referred to as inverse problems, since having the observations and the forward-model between the observations and the sources is generally not enough to obtain solutions to these problems directly, due to their ill-posed nature.

Indeed, unlike the forward-model which is well-posed every time (cf. Hadamard \cite{hadamard2003lectures}), inverse problems are generally ill-posed \cite{zhu2011inverse}. Therefore, dealing with the prior knowledge about the object of interest plays a crucial role in reaching a stable/unique solution. This leads to regularisation based methods, which received great attention hitherto in the literature \cite{aggarwal2006line, cetin2014sparsity, selesnick1, selesnick2017total, karakucs2019ship2, zuo2013generalized, mohammad2012bayesian, anantrasirichai2017line}.

In most of these examples, the common choice of regularisation functions is based on the $L_1$ norm, due to its convexity and capability to induce sparsity effectively. Another important example of a convex regularisation function is the total variation ($TV$) norm. It constitutes the state-of-the-art in denoising applications, due to its efficiency in smoothing. Despite their common usage, the $L_1$ norm penalty tends to underestimate high-amplitude/intensity values, whilst $TV$ tends to over-smooth the data and may lead to loss of details. Non-convex penalty functions can generally lead to better and more accurate estimations \cite{selesnick2, nikolova2005analysis, chen2014convergence} when compared to $L_1$, $TV$, or some other convex penalty functions. Notwithstanding this, due to the non-convexity of the penalty functions, the overall cost function becomes non-convex, which implies a multitude of sub-optimal local minima.

Convexity preserving non-convex penalty functions are thus essential, the idea having been successfully applied by Blake, Zimmerman \cite{blake1987visual}, and Nikolova \cite{nikolova2010fast}, and further developed in \cite{selesnick2017total,selesnick2,parekh2015convex,lanza2016convex,malek2016class,selesnick2020non,anantrasirichai2020image}. Specifically, a convex denoising scheme is proposed with tight frame regularisation in \cite{parekh2015convex}, whilst \cite{lanza2016convex} proposes the use of parameterised non-convex regularisers to effectively induce sparsity of the gradient magnitudes. In \cite{selesnick2017total}, the Moreau envelope is used for TV denoising in order to preserve the convexity of a TV denoising cost function. The non-convex generalised minimax concave (GMC) penalty function is proposed in \cite{selesnick2} for convex optimisation problems.

Another important reason behind the appeal of the aforementioned penalty functions in applications, is the existence of closed-form expressions for their proximal operators. Specifically, the proximal operator of a regularisation function has been introduced in conjunction with inverse problems, to help solving various signal processing tasks. Proximal operators are powerful and flexible tools with attractive properties, which enable solutions to non-differentiable optimisation problems, and make them suitable for iterative minimisation algorithms \cite{combettes2011proximal}, such as forward-backward (FB), or the alternating direction method of multipliers (ADMM). Remarkably, many widespread regularisation functions have corresponding proximal operators available in closed form, or at least numerical methods to calculate them exist. For example, the soft thresholding function is the proximal operator for the $L_1$ norm, whereas the proximal operator is the generalised soft thresholding (GST) \cite{zuo2013generalized} for $L_p$ norm penalty function. It is efficiently computed by using Chambolle's method \cite{chambolle2004algorithm} for $TV$ norm, whilst the GMC penalty only necessitates the use of soft-thresholding, or firm-thresholding in the case of diagonal forward operator $\mathcal{A}$ as shown in \cite{selesnick2}.

The quest for finding the most appropriate penalty function, eventually in relation to an explicit prior distribution characterising the data statistics, is far from being over. In this work, we consider the Cauchy distribution, a special member of the $\alpha$-stable distribution family, which is known for its ability to model heavy-tailed data in various signal processing applications. As a prior in image processing applications, it behaves as a sparsity-enforcing  one, similar to $L_1$ and $L_p$ norms \cite{mohammad2012bayesian}. It has already been used in denoising applications by modelling sub-band coefficients in transform domains \cite{achim2004image, bhuiyan2007spatially, chen2008wavelet, ranjani2010dual, gao2013directionlet}. Moreover, the Cauchy distribution was also used as a noise model in image processing applications, by employing it for the data fidelity term in combination with quadratic \cite{sciacchitano2015variational} and TV norm \cite{mei2018cauchy} based penalty terms. 

The general approach involves the use of a variational Bayesian methodology to solve Cauchy regularised inverse problems due to its lack of a closed-form proximal operator. This prevents the Cauchy prior from being used in proximal splitting algorithms such as the FB, and ADMM. Moreover, having a proximal operator would also make the Cauchy based regularisation function applicable in advanced Bayesian signal/image processing methods, such as in uncertainty quantification (UQ) via proximal Markov Chain Monte Carlo ($p$-MCMC) algorithms \cite{cai2018uncertainty, durmus2018efficient}.

In this paper, we propose a proximal splitting methodology for solving inverse problems of the form
\begin{align}\label{equ:IP}
y = \mathcal{A}x + n,
\end{align}
where $y \in \mathbb{R}^M$ denotes the observation (can be either an image or some other kind of signals), $x \in \mathbb{R}^N$ is the unknown signal, which can also be referred to as target data (either an enhanced data or the raw data), $\mathcal{A} \in \mathbb{R}^{M\times N}$ is the forward model operator and $n \in \mathbb{R}^M$ represents the additive noise.
Specifically, we use a non-convex penalty function based on the Cauchy distribution, in order to capture the heavy-tailed and/or sparse characteristics of the target, $x$, with a number of original contributions, which include:
\begin{enumerate}
    \item deriving a closed form expression for the Cauchy proximal operator inspired by \cite{wan2011segmentation}, which makes Cauchy regularisation applicable in proximal splitting algorithms.
    \item deriving the condition that guarantees convergence of the Cauchy proximal operator to the global minimum. Even though the proposed Cauchy based penalty function is non-convex, the overall optimisation problem remains convergent either (i) through the use of proximal splitting algorithms, or (ii) through convexity of the cost function itself when the forward operator $\mathcal{A}$ satisfies the assumptions of orthogonality or of being an over-complete tight frame.
\end{enumerate}

We investigate the performance of the proposed Cauchy-based penalty function in comparison to $L_1$ and $TV$ norm penalty functions in three examples including (i) 1D signal denoising, (ii) 2D image restoration
, and (iii) error recovery in multiple-input-multiple-output (MIMO) signal detection. Furthermore, we also study the effects of following/violating the proposed convergence conditions in those same examples. Please note that the present paper focuses on the theoretical aspects involving the use of the non-convex Cauchy-based penalty function in solving inverse problems and only presents three generic, but illustrative, signal/image processing examples. Further performance analysis of the proposed penalty function is investigated in \cite{karakucs2019cauchy2}, for four different synthetic aperture radar (SAR) imaging inverse problems. In addition, its merits in line artefacts quantification for lung ultrasound images of COVID-19 patients is investigated in \cite{karakucs2020covid19}. 

The rest of the paper is organised as follows: Section \ref{sec:SARIP} presents the proposed Cauchy proximal operator. Convergence analysis of the proposed method is given in Section \ref{sec:ConvAn} along with the corresponding Cauchy proximal splitting method. In Section \ref{sec:results}, the experimental validation on the proposed conditions, and an analysis on 1D, 2D and MIMO inverse problems are presented as well as a discussion on further applications in \cite{karakucs2019cauchy2,karakucs2020covid19}. We conclude our study and describe future work directions in Section \ref{sec:conc}.

\section{The Cauchy Proximal Operator}\label{sec:SARIP}
Recalling the generic signal model in (\ref{equ:IP}), a stable solution to this ill-posed inverse problem is obtained through an optimisation of the following form:
\begin{align}\label{equ:IP2}
    \hat{x} = \arg \min_x \bigg\{F(x) = \Psi(y, \mathcal{A}x) + \psi(x) \bigg\}
\end{align}
where $F: \mathbb{R}^N \rightarrow \mathbb{R}$ is the cost function to be minimised, $\Psi: \mathbb{R}^N \rightarrow \mathbb{R}$ is a function which represents the data fidelity term and $\psi: \mathbb{R}^N \rightarrow \mathbb{R}$ is the regularisation function (the penalty term). Under the assumption of an independent and identically distributed (iid) Gaussian noise, the data fidelity term can be expressed as
\begin{align}\label{equ:LHD}
    \Psi(y, \mathcal{A}x) = \frac{\|y - \mathcal{A}x\|_2^2}{2\sigma^2}
\end{align}
where $\sigma$ refers to the standard deviation of the noise level. Based on a prior probability density function (pdf) $p(x)$, the problem of estimating $x$ from the noisy observation $y$ by using the signal model in (\ref{equ:IP}) turns into the following minimisation problem in a variational framework
\begin{align}\label{equ:mini1}
    \hat{x} = \arg\min_x \frac{\|y - \mathcal{A}x\|_2^2}{2\sigma^2} - \log p(x)
\end{align}
where we define the penalty function $\psi(x)$ as the negative logarithm of the prior distribution $-\log p(x)$. The selection of $\psi(x)$ (or equivalently $p(x)$) plays a crucial role in estimating $x$ in order to overcome the ill-posedness of the problem and to obtain a stable/unique solution. In the literature, depending on the application, the penalty term $\psi(x)$ has various forms, such as $L_1$, $L_2$, $TV$ or $L_p$ norms, to name but a few possible choices.



In this study, we propose the use of a penalty function based on the Cauchy distribution, which is known to be heavy-tailed and promote sparsity in various applications. From a purely theoretical viewpoint, our preference for the Cauchy model over other candidate models stems from its membership of the $\alpha$-Stable family of distributions. Specifically, unlike other empirical distributions able to faithfully fit distributions with heavy-tails, $\alpha$-stable distributions are motivated by the generalised central limit theorem (CLT) similarly to the way Gaussian distributions are motivated by the classical CLT.
However, although the (symmetric) $\alpha$-stable density behaves approximately like a Gaussian
density near the origin, its tails decay at a lower rate than the Gaussian
density tails. Indeed, let $X$ be a non-Gaussian symmetric $\alpha$-Stable random variable. Then, as $x\rightarrow\infty$
\begin{align}
P(X>x) \; \sim \; c_\alpha x^{-\alpha}
\label{eq:tails}
\end{align}
where $c_\alpha = \Gamma(\alpha) (\sin\frac{\pi\alpha}{2})/\pi$,
$\Gamma(x)=\int_{0}^{\infty}t^{x-1}e^{-t}\,dt$ is the Gamma function,
and the statement $h(x)\sim g(x)$ as $x\rightarrow\infty$ means that
$\lim_{x\rightarrow\infty} h(x)/g(x) =1$. Hence, the
tail probabilities are asymptotically power laws. We should note that expression (\ref{eq:tails}) gives exactly
the tail probability of the Pareto distribution, and hence the term {\em ``stable
Paretian laws''} is sometimes used to distinguish between the fast decay of the
Gaussian law and the Pareto like tail behaviour when $\alpha<2$.

Contrary to the general $\alpha$-stable family, the Cauchy distribution has a closed-form probability density function, which is defined by
\begin{align} \label{equ:Cauchy}
    p(x) \propto \frac{\gamma}{\gamma^2+x^2}
\end{align}	
where $\gamma$ is the dispersion (scale) parameter, which controls the spread of the distribution. By replacing $p(x)$ in (\ref{equ:mini1}) with the Cauchy prior given in (\ref{equ:Cauchy}), we obtain the following optimisation problem
\begin{align}\label{equ:miniCauchy}
    \hat{x}_{\text{Cauchy}} = \arg\min_x \frac{\|y - \mathcal{A}x\|_2^2}{2\sigma^2} - \sum_{i,j}\log\left(\frac{\gamma}{\gamma^2+x_{ij}^2}\right).
\end{align}

Using proximal splitting methods has numerous advantages when compared to classical methods. In particular, they (i) work under general conditions, e.g. for functions which are non-smooth and extended real-valued, (ii) generally have simple forms, so they are easy to derive and implement, (iii) can be used in large scale problems. In addition, most of the proximal splitting algorithms are generalisations of the classical approaches such as the projected gradient algorithm \cite{parikh2014proximal}.

In order to solve the minimisation problem in (\ref{equ:miniCauchy}) through efficient proximal algorithms such as forward-backward (FB) or the alternating direction of multipliers method (ADMM), \textit{the proximal operator} of the Cauchy regularisation function should be defined.
Proximal operators have been extensively used in solving inverse problems, whereby they can generally be computed efficiently using various algorithms for a given regularisation function, e.g. the soft thresholding function for $L_1$ norm, or Chambolle's method for the $TV$ norm \cite{chambolle2004algorithm}. Besides, $prox_h^{\mu}$ has similar properties to the gradient mapping operators, which point in the direction of the minimum of $h$.
Thus, for any function $h(\cdot)$ and $\mu>0$, the proximal operator, $prox_h^{\mu}:\mathbb{R}\rightarrow\mathbb{R}$ is defined as \cite{combettes2011proximal, parikh2014proximal}
\begin{align}\label{equ:prox}
    prox_h^{\mu}(x) = \arg\min_u \left\{h(u) + \|x - u \|_2^2 / 2\mu \right\}.
\end{align}
For a Cauchy based penalty function, we recall that the function $h(\cdot)$ is given by
\begin{align}\label{equ:neglogCuachy}
h(x) = -\log\left(\dfrac{\gamma}{\gamma^2+x^2}\right),
\end{align}
which implies the Cauchy proximal operator is
\begin{align}\label{equ:proxCauchy}
    prox_{Cauchy}^{\mu}(x) = \arg\min_u \left\{\frac{\|x - u \|_2^2}{2\mu} - \log\left(\frac{\gamma}{\gamma^2+u^2}\right) \right\}
\end{align}

The solution to this minimisation problem can be obtained by taking the first derivative of (\ref{equ:proxCauchy}) in terms of $u$ and setting it to zero. Hence we have
\begin{align}\label{equ:proxCauchy2}
    u^3-xu^2+(\gamma^2+2\mu)u-x\gamma^2 = 0.
\end{align}
Wan et al. \cite{wan2011segmentation} proposed a Bayesian maximum a-posteriori (MAP) solution to the problem of denoising a Cauchy signal in Gaussian noise, and referred to this solution as “Cauchy shrinkage”. Similarly, the minimisation problem in (\ref{equ:proxCauchy}) can be solved with the same approach as in \cite{wan2011segmentation}, using however a different parameterisation.

Hence, following \cite{wan2011segmentation}, the solution to the cubic function given in (\ref{equ:proxCauchy2}) can be obtained through Cardano’s method, which is given in Algorithm \ref{alg:proxCauchy}.

\begin{algorithm}[ht!]
\caption{The Cauchy Proximal Operator}\label{alg:proxCauchy}
\setstretch{1.2}
\begin{algorithmic}[1]
\Procedure{proxCauchy}{$x, \gamma, \mu$}
    \State $p \gets \gamma^2 + 2\mu - \frac{x^2}{3}$
    \State $q \gets x\gamma^2 + \frac{2x^3}{27} - \frac{x}{3}\left(\gamma^2 + 2\mu\right)$
    \State $s \gets \sqrt[3]{q/2 + \sqrt{p^3/27 + q^2/4}}$
    \State $t \gets \sqrt[3]{q/2 - \sqrt{p^3/27 + q^2/4}}$
    \State $z \gets \frac{x}{3} + s + t$ \Comment{Cauchy proximal operator result}
    \State \textbf{return} $z$
\EndProcedure
\end{algorithmic}
\end{algorithm}

Figure \ref{fig:normproxComp}-(a) depicts the behaviour of five different penalty functions. Among those, the proposed Cauchy based penalty is a non-convex function as are the $L_p$ and GMC penalties, and contrary to the $L_1$ and $L_2$ norms, which are convex. Figure \ref{fig:normproxComp}-(b) illustrates the behaviour of the Cauchy proximal operator compared to the soft ($L_1$), hard ($L_0$) and firm (GMC) thresholding functions. Examining Figure \ref{fig:normproxComp}-(b), the Cauchy proximal operator can be regarded as a compromise between soft and hard thresholding functions.

\begin{figure}[ht]
\centering
\subfigure[]{\includegraphics[width=.6\linewidth]{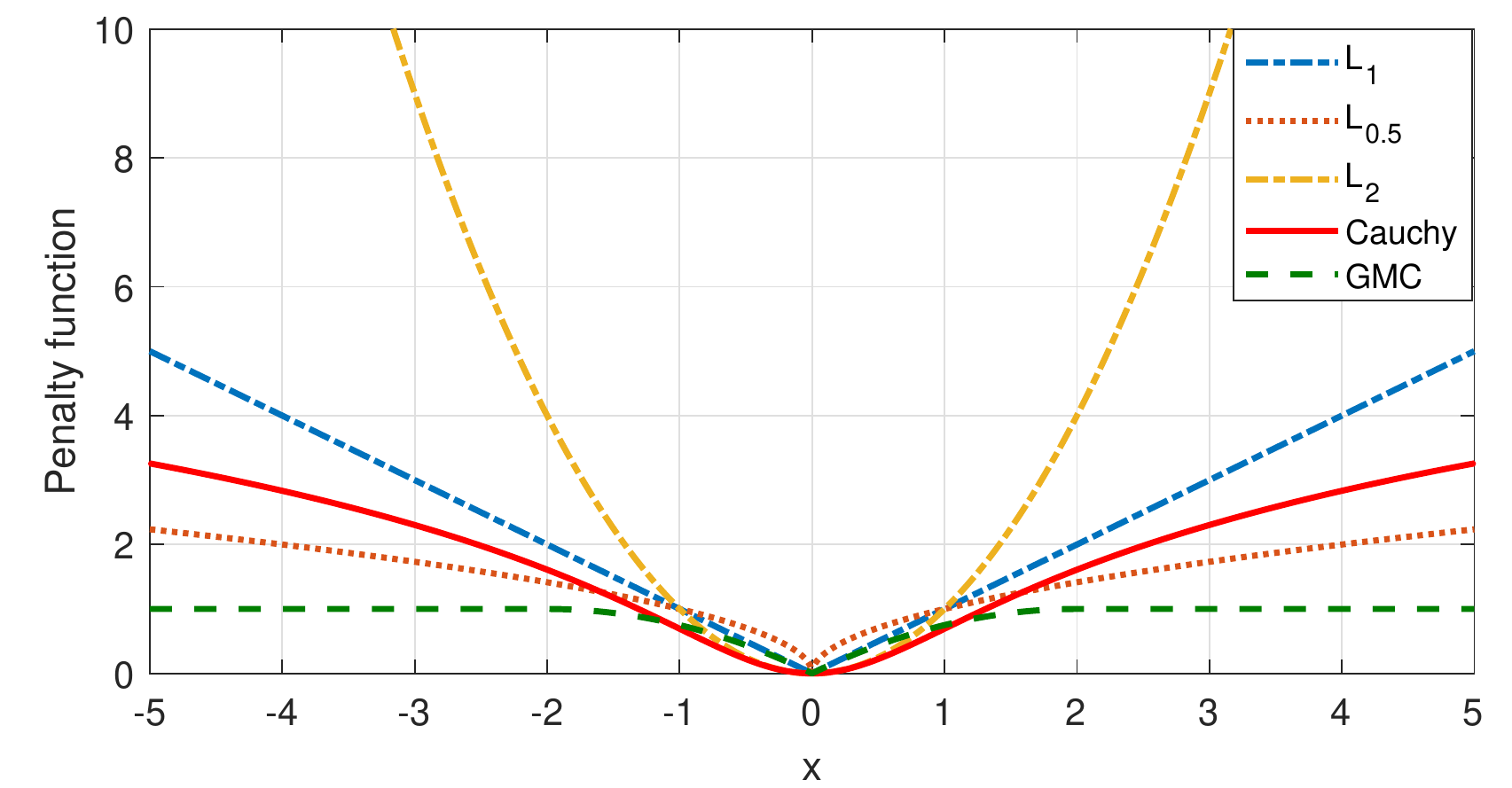}}
\subfigure[]{\includegraphics[width=.6\linewidth]{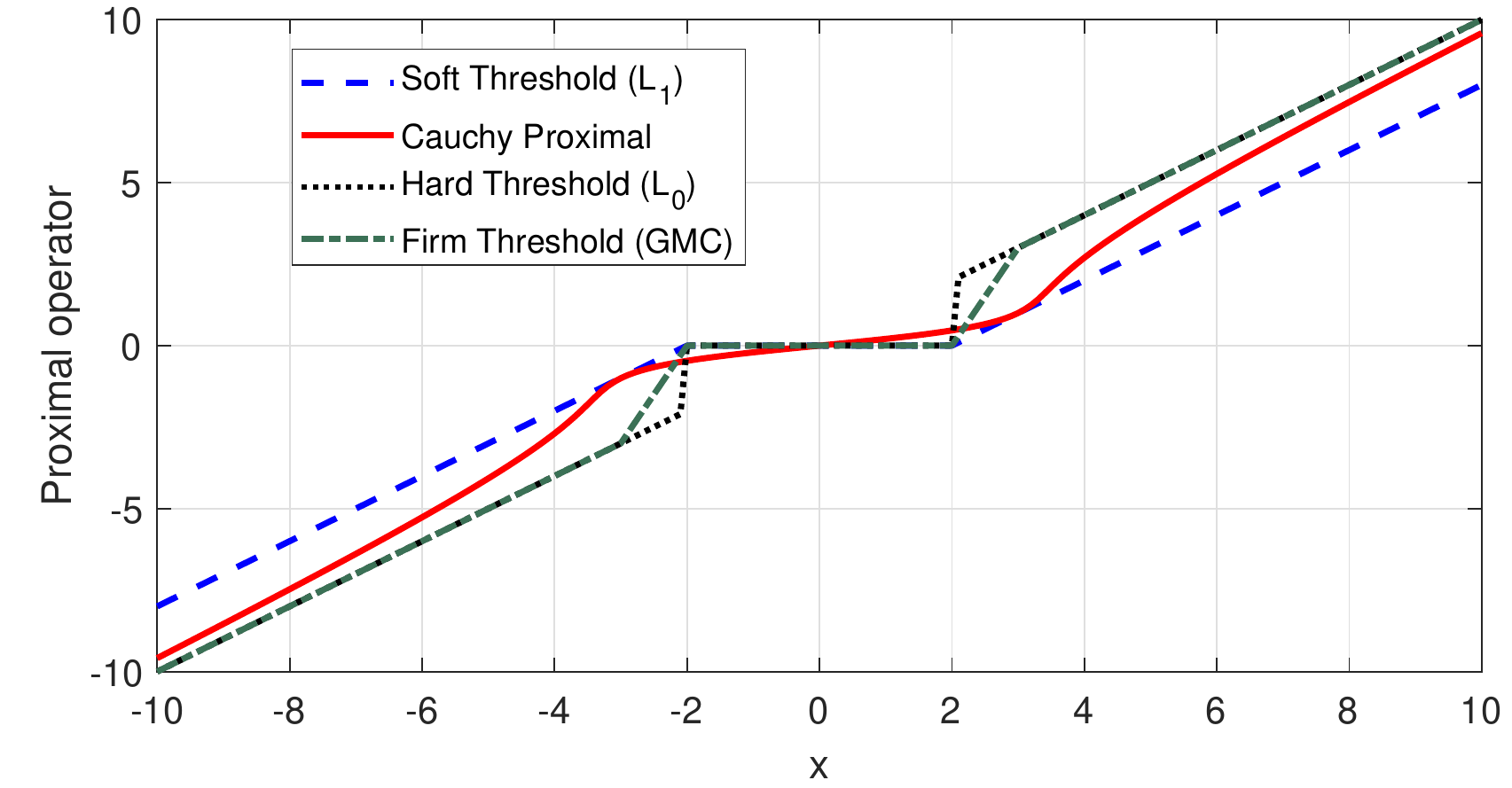}}
\caption{Visual comparison for (a) penalty functions, (b) proximal operators.}
\label{fig:normproxComp}
\end{figure}

\section{Convergence analysis}\label{sec:ConvAn}
In order to analyse the convergence properties of the proposed method, we start from the minimisation given in (\ref{equ:miniCauchy}). Since we have a quadratic data fidelity term and a non-convex penalty function, the overall cost function in (\ref{equ:miniCauchy}) might be non-convex. To benefit from convex optimisation principles in solving (\ref{equ:miniCauchy}), we seek to ensure that the cost function in (\ref{equ:miniCauchy}) is convex by controlling the general system parameters e.g. $\sigma$ and $\gamma$. For this purpose, we start with the following lemma.

\begin{lem}\label{lem:lemma1}
Let the function $h$ has the form defined in (\ref{equ:neglogCuachy}) with $\gamma>0$. The function $h(x)$ is twice continuously differentiable and non convex $\forall x$.
\end{lem}

\begin{proof}
In order to prove that the function $h(x)$ is twice continuously differentiable, we need to show: $h'(0^+)=h'(0^-)$ and $h''(0^+)=h''(0^-)$. Thus, we start with the first derivative: $h'(x) = \frac{2x}{\gamma^2 + x^2}$, whose zero limit values around $x=0$ are $h'(0^+)=h'(0^-)=0$. Similarly, the second derivative is obtained as: $h''(x) = \frac{2\gamma^2 - 2x^2}{x^4 + 2\gamma^2x^2 + \gamma^4}$. Around 0, limit values for the second derivative are $h''(0^+)=h''(0^-)=\frac{2}{\gamma^2}$. Thus, the function $h$ is obviously twice continuously differentiable.

The function $h$ is convex if $h''(x) \geq 0$, $\forall x$. However, we recall that the second derivative is $h''(x) = \frac{2\gamma^2 - 2x^2}{x^4 + 2\gamma^2x^2 + \gamma^4}$, which satisfies $h''(x) \geq 0$ only for $-\gamma\leq x \leq \gamma$ and thus, $h$ is not convex.
\end{proof}

\begin{figure}[h!]
    \centering
    \includegraphics[width=0.6\linewidth]{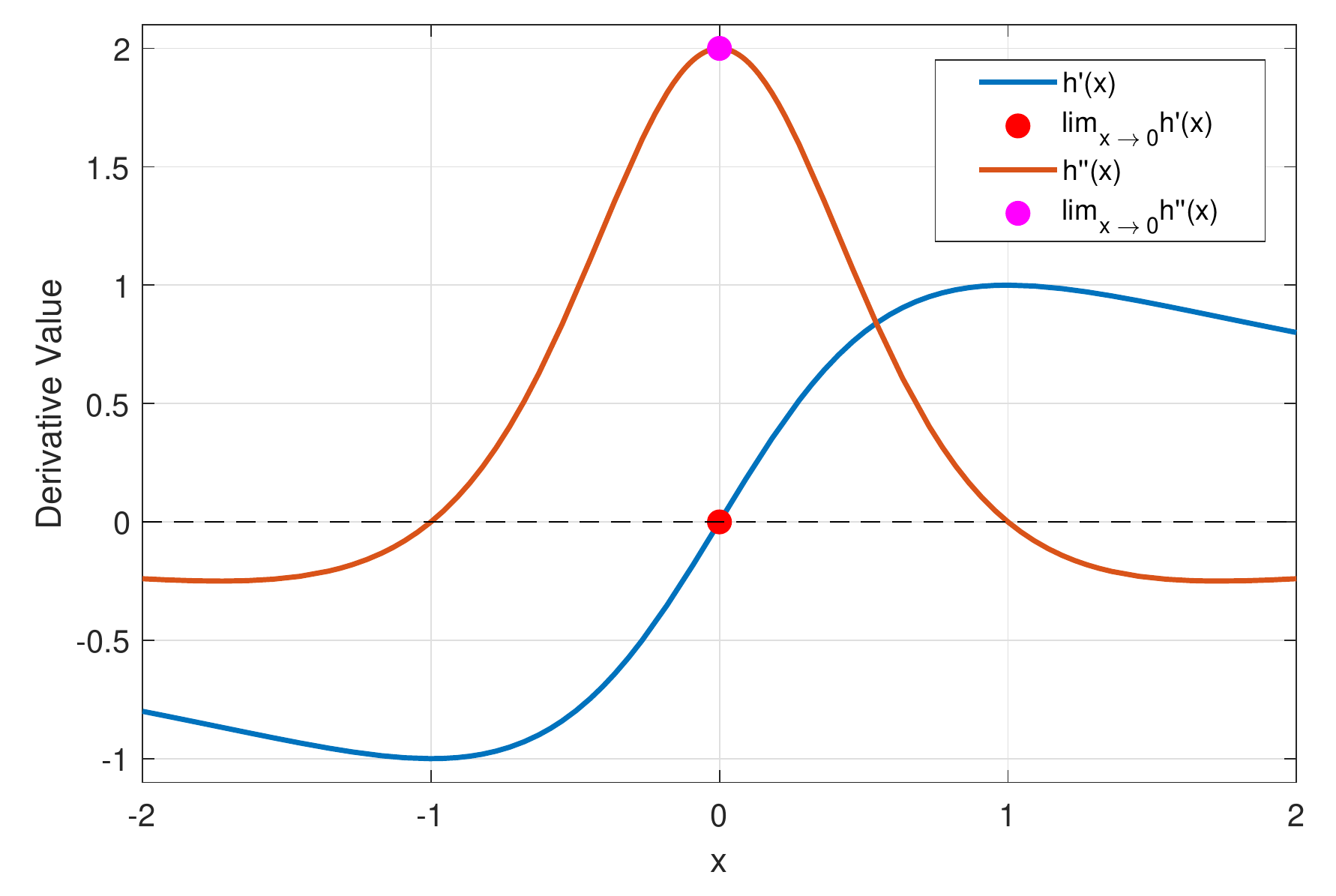}
    \caption{Derivative of the Cauchy-based penalty for $\gamma = 1$.}
    \label{fig:lemma}
\end{figure}

\begin{rem}
The function $h$ from Lemma \ref{lem:lemma1} is non-convex except for $-\gamma\leq x \leq \gamma$, wherein $h''(x) \geq 0$. Since $\gamma$ generally takes relatively small values when compared to $x$, it is not practical to enforce this condition for convexity. Therefore, we assume that the function $h$ is non-convex almost everywhere on the support of $x$.
\end{rem}

Figure \ref{fig:lemma} depicts plots of both the first and second derivatives of the function $h$ with $\gamma = 1$. It offers a graphical confirmation of the proof to Lemma \ref{lem:lemma1}. Specifically, red and magenta dots in Figure \ref{fig:lemma} show limit values for the first and second derivatives, respectively. Besides, the horizontal dashed-line shows derivative value equals to zero, where the second derivative takes negative values outside of the interval $-\gamma\leq x \leq \gamma$, which demonstrates the non-convexity of the function $h$.

We now state the following theorem that establishes the condition to preserve the convexity of the cost function in (\ref{equ:miniCauchy}).

\begin{thm}\label{thm:theorem1_new}
Let $h$ be the twice continuously differentiable and non-convex penalty function in (\ref{equ:neglogCuachy}) with $\gamma>0$, and the forward operator $\mathcal{A}$ either orthogonal satisfying $\mathcal{A}^T\mathcal{A}=\mathcal{I}$, or an overcomplete tight frame satisfying $\mathcal{A}^T\mathcal{A} \approx r\mathcal{I}$ with $r>0$ where $\mathcal{I}$ is the identity matrix. Then, the cost function $F:\mathbb{R}^N\rightarrow\mathbb{R}$
\begin{align}
    F(x) = \frac{\|y - \mathcal{A}x \|^2}{2\sigma^2}  -\sum_{i,j}\log\left(\frac{\gamma}{\gamma^2 + x_{ij}^2} \right)
\end{align}
is strictly convex if
\begin{align}
    \gamma \geq \frac{\sigma}{2\sqrt{r}}.
\end{align}
\end{thm}
\begin{proof}
According to Lemma \ref{lem:lemma1}, the function $F$ is twice continuously differentiable, and we further express the Hessian of $F$ as
\begin{align}
    \bigtriangledown^2F(x) = \dfrac{\mathcal{A}^T\mathcal{A}}{\sigma^2} + \dfrac{2\gamma^2 - 2x^2}{x^4 + 2\gamma^2x^2 + \gamma^4}
\end{align}
This must be positive definite in order for the cost function $F$ to be convex with $\bigtriangledown^2F(x) \succeq 0$, then we write
\begin{align}
    \label{equ:thmFHess11}\dfrac{\mathcal{A}^T\mathcal{A}}{\sigma^2} + \dfrac{2\gamma^2 - 2x^2}{x^4 + 2\gamma^2x^2 + \gamma^4} &\succeq 0.
\end{align}
Recalling that $\mathcal{A}^T\mathcal{A} \approx r\mathcal{I}$, then we have
\begin{align}
\dfrac{r\mathcal{I}}{\sigma^2} + \dfrac{2\gamma^2 - 2x^2}{x^4 + 2\gamma^2x^2 + \gamma^4} &\succeq 0\\
    \dfrac{rx^4 + r2\gamma^2x^2 + r\gamma^4\  + 2\sigma^2\ (\gamma^2 - x^2)}{x^4 + 2\gamma^2x^2 + \gamma^4} &\succeq 0,\\
    rx^4 - 2r\gamma^2x^2 + r\gamma^4\  + 2\sigma^2\gamma^2\  - 2\sigma^2x^2 &\succeq 0,\\
    \label{equ:subthm1}rx^4 + 2\sqrt{r}x^2\left(\sqrt{r}\gamma^2 - \frac{\sigma^2}{\sqrt{r}}\right) + r\gamma^4\  + 2\sigma^2\gamma^2\  &\succeq 0.
\end{align}
To complete the square on the left-hand side, we add and subtract $\frac{\sigma^4}{r}\ $ and $4\sigma^2\gamma^2$ to (\ref{equ:subthm1}). Then, we have
\begin{multline}
    rx^4 + 2\sqrt{r}x^2\left(\sqrt{r}\gamma^2 - \frac{\sigma^2}{\sqrt{r}}\right) + \left(\sqrt{r}\gamma^2 - \frac{\sigma^2}{\sqrt{r}}\right)^2\ \\ -\left(\frac{\sigma^4}{r} + 4\sigma^2\gamma^2\right)\
     \succeq 0,
\end{multline}
\begin{multline}
    \label{equ:proof11}\left(\sqrt{r}x^2 + (\sqrt{r}\gamma^2 - \frac{\sigma^2}{\sqrt{r}})\right)^2 + \sigma^2(4\gamma^2 - \frac{\sigma^2}{r})\  \succeq 0.
\end{multline}
It can be easily seen that the term $\left(\sqrt{r}x_i^2 + (\sqrt{r}\gamma^2 - \frac{\sigma^2}{\sqrt{r}})\right)^2$ is always positive as well as the noise standard deviation $\sigma$. Thus, for the inequality in (\ref{equ:proof11}) to hold,  the simplified condition of
\begin{align}
    4\gamma^2 - \frac{\sigma^2}{r} \geq 0
\end{align}
should be satisfied. This leads to the condition required to ensure (strict) convexity of the function $F$:
\begin{align}\label{equ:ProofCond11}
    \gamma \geq \frac{\sigma}{2\sqrt{r}}.
\end{align}
and the existance of a unique solution for the given cost function.
\end{proof}

Theorem \ref{thm:theorem1_new} provides the critical value for the scale parameter of the non-convex Cauchy-based penalty that ensures the whole cost function remains convex. As noted, this condition depends of the value of the noise standard deviation $\sigma$ and the parameter $r$, which follows from the assumption that $\mathcal{A}^T \mathcal{A}$ has a diagonal form. In the following we make another remark.

\begin{rem}
Preserving the convexity of the problem overall, in spite of the non-convexity of the Cauchy based penalty function, requires to have a forward operator $\mathcal{A}$, which is orthonormal ($\mathcal{A}^T\mathcal{A}=\mathcal{I}$) or constitutes an overcomplete tight frame with $\mathcal{A}^T\mathcal{A} \approx r\mathcal{I}$. For applications such as denoising, where $\mathcal{A}=\mathcal{I}$ and situations where $\mathcal{A}$ is the Fourier or orthogonal wavelet transform, convergence is guaranteed according to Theorem \ref{thm:theorem1_new}. However, in cases where forward models do not satisfy the relation $\mathcal{A}^T\mathcal{A} \approx r\mathcal{I}$, or estimating $r$ is challenging, the condition given in Theorem \ref{thm:theorem1_new} will not be suitable to ensure convergence.
\end{rem}

For more general situations,which include the assumptions in Theorem \ref{thm:theorem1_new} and beyond, we propose another solution which guarantees convergence provided that the solution is obtained via a proximal splitting algorithm even though $\mathcal{A}^T\mathcal{A} \neq r\mathcal{I}$.
We start by another lemma, which states a condition to ensure that the Cauchy proximal operator cost function is convex, and converges to a global minimum even though it corresponds to a non-convex penalty function.

\begin{lem} \label{lem:lemma2}
The function $J:\mathbb{R} \rightarrow \mathbb{R}$
\begin{align}
    J(u) = \frac{\|x - u \|^2}{2\mu}  -\log\left(\frac{\gamma}{\gamma^2 + u^2} \right)
\end{align}
with $\gamma>0$, $\mu>0$,
is strictly convex if the following condition is obeyed:
\begin{align}
    \gamma \geq \frac{\sqrt{\mu}}{2}.
\end{align}
\end{lem}

\begin{proof}
We first express the second derivation of $J$ as
\begin{align}
    J''(u) = \dfrac{1}{\mu} + \dfrac{2\gamma^2 - 2u^2}{u^4 + 2\gamma^2u^2 + \gamma^4}.
\end{align}
Then, akin to the proof of Theorem \ref{thm:theorem1_new}, we continue with the convexity condition
\begin{align}
    J''(u) &\geq 0,\\
    \label{equ:thmFHess}\dfrac{1}{\mu} + \dfrac{2\gamma^2 - 2u^2}{u^4 + 2\gamma^2u^2 + \gamma^4} &\geq 0,\\
    \dfrac{u^4 + 2\gamma^2u^2 + \gamma^4 + 2\mu(\gamma^2 - u^2)}{u^4 + 2\gamma^2u^2 + \gamma^4} &\geq 0,\\
    u^4 - 2u^2(\mu - \gamma^2) + \gamma^2(\gamma^2 + 2\mu) &\geq 0.
\end{align}
To complete the square on the left-hand side, we add and subtract $\mu^2$ and $4\gamma^2\mu$, which gives
\begin{align}
    (u^2)^2 - 2(u^2)(\mu - \gamma^2) + (\mu - \gamma^2)^2 - \mu^2 + 4\gamma^2\mu &\geq 0\\
    \label{equ:proof}(u^2 - (\mu - \gamma^2))^2 + \mu(4\gamma^2 - \mu) &\geq 0.
\end{align}

Since the term $(u^2 - (\mu - \gamma^2))^2$ is always positive as well as the step size $\mu$, for the inequality in (\ref{equ:proof}) to hold, the condition
\begin{align}
    4\gamma^2 - \mu \geq 0.
\end{align}
should be satisfied. Hence, the cost function in the Cauchy proximal operator $J$ becomes strictly convex if
\begin{align}\label{equ:ProofCond}
    \gamma \geq \frac{\sqrt{\mu}}{2}.
\end{align}
\end{proof}

In Figure \ref{fig:Fcomp}, we demonstrate the effect of the relationship between $\mu$ and $\gamma$ on $J(u)$ and its second derivative $J''(u)$. Both sub-figures in Figure \ref{fig:Fcomp} obviously show that violating the expression for convexity given in Lemma \ref{lem:lemma2}, makes the cost function non-convex.


\begin{figure}[t]
\centering
\subfigure[]{\includegraphics[width=.6\linewidth]{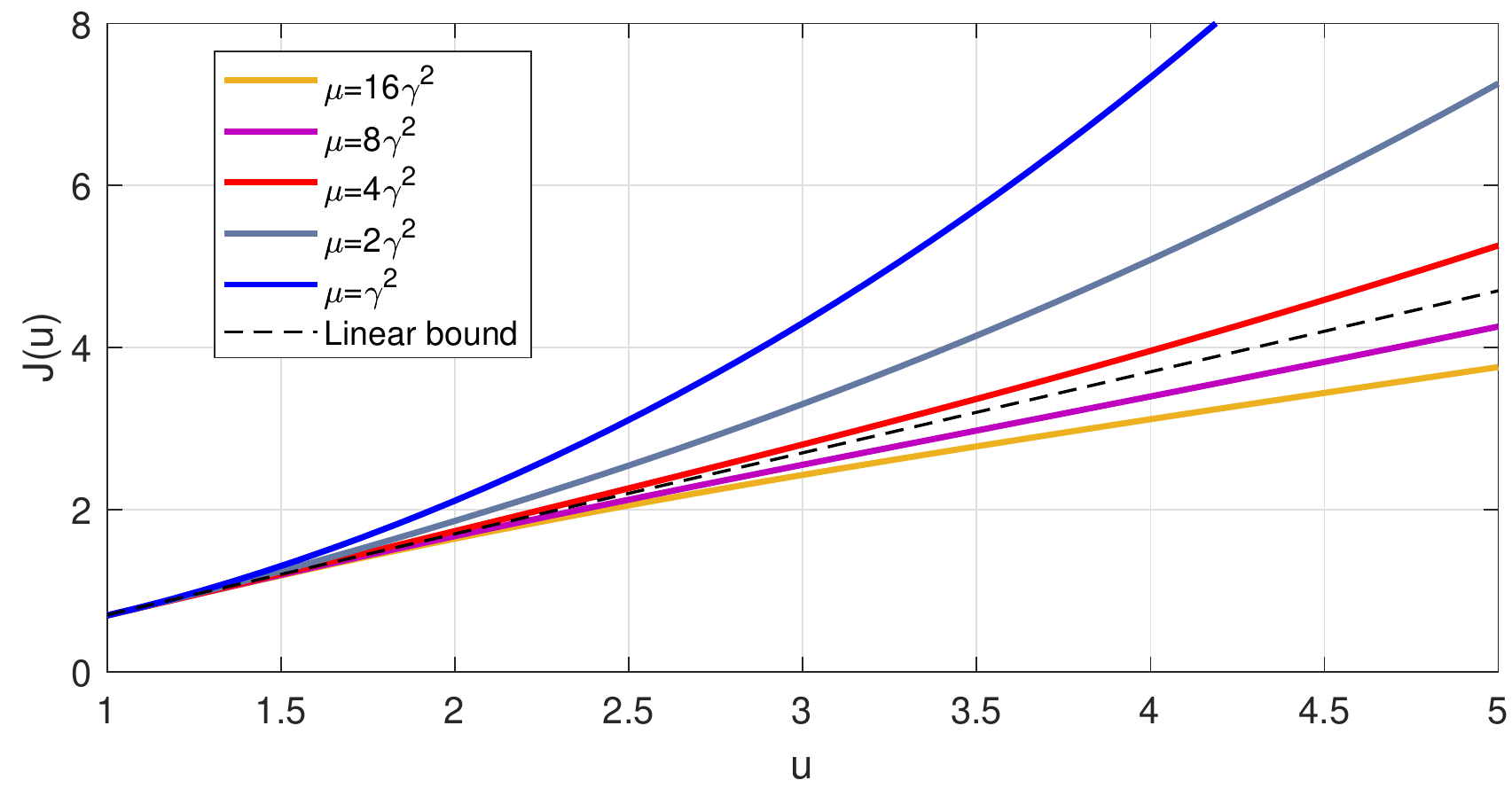}}
\subfigure[]{\includegraphics[width=.6\linewidth]{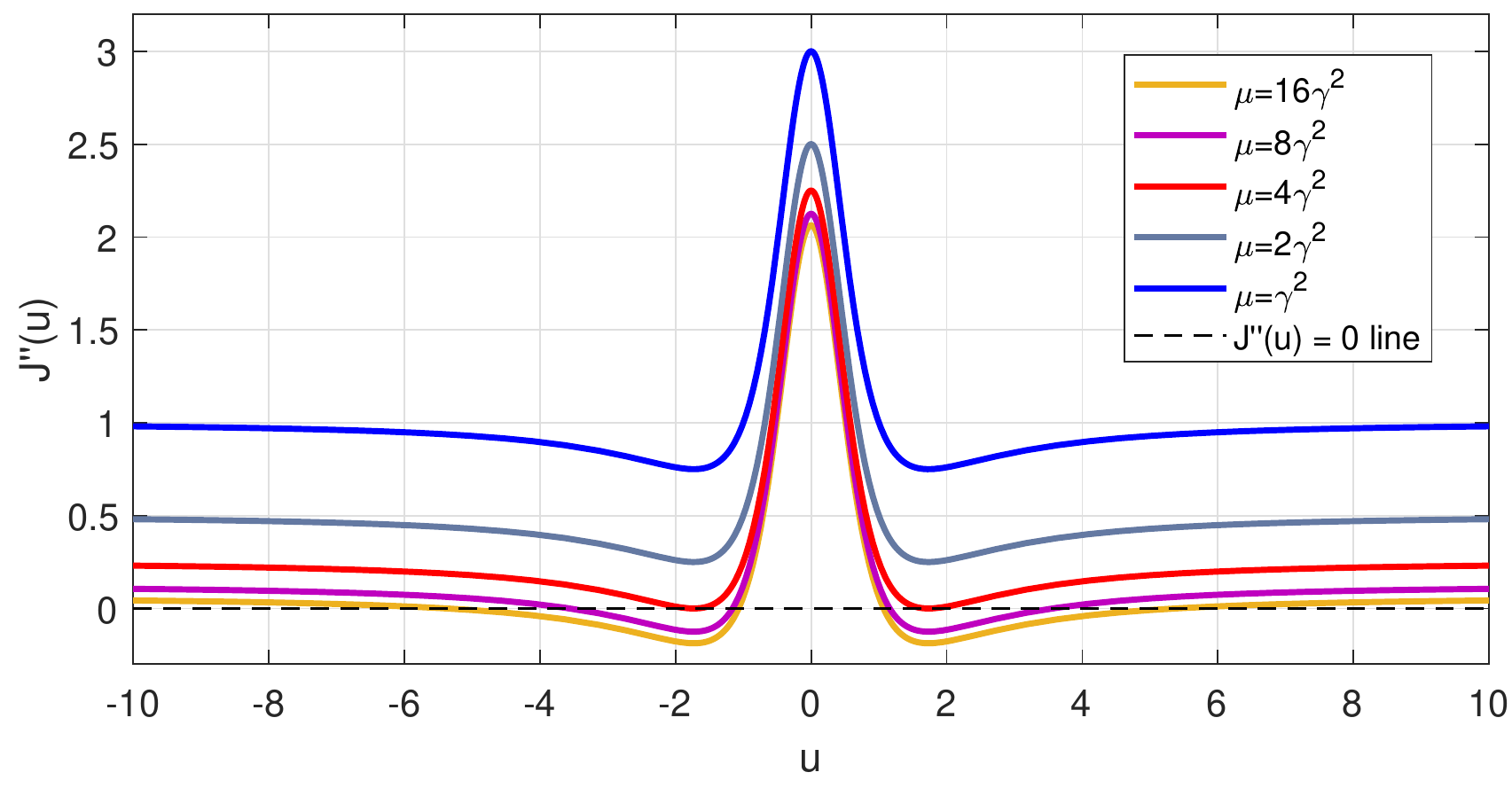}}
\caption{Behaviour of $J(u)$ (a) and $J''(u)$ (b) for various step sizes $\mu$ given a fixed value of the scale parameter $\gamma$.}
\label{fig:Fcomp}
\end{figure}

\begin{rem}
Instead of providing a condition to ensure that the Cauchy based cost function remains convex, Lemma \ref{lem:lemma2} provides a condition which preserves the convexity of the Cauchy proximal operator. Please note that a solution to the proximal operator $prox_{Cauchy}^{\mu}$ can always be computed since it has an explicit expression which is given in Algorithm \ref{alg:proxCauchy}. However, the convexity condition given in Lemma \ref{lem:lemma2} leads on to the theorem in the following section, which provides the required condition to guarantee the convergence for the cost function in (\ref{equ:miniCauchy}), when relaxing the assumption of orthogonality and over-completeness of the forward operator $\mathcal{A}$ in Theorem \ref{thm:theorem1_new}.
\end{rem}

There are several proximal splitting algorithms that can be used to solve the optimisation problem in (\ref{equ:mini1}), including the forward-backward splitting, Douglas-Rachford (DR) splitting, or alternating direction method of multipliers (ADMM) \cite{combettes2011proximal} to name but a few. In this paper, we focus on the forward-backward algorithm to obtain efficient solutions to the inverse problem in (\ref{equ:IP}).
Indeed, an optimisation problem where the cost function $F(x)$ is split into two functions, with the form
\begin{align}\label{equ:FB1}
    \arg\min_x \left[ F(x) = (f_1 + f_2)(x)\right]
\end{align}
can be solved via the FB algorithm. Provided $f_2:\mathbb{R}^N \rightarrow \mathbb{R}$ is $L$-Lipchitz differentiable with Lipchitz constant $L$ and $f_1:\mathbb{R}^N \rightarrow \mathbb{R}$, then, given the signal at $n^{\text{th}}$ iteration (i.e. $x^{(n)}$), the optimisation problem in (\ref{equ:FB1}) is solved iteratively as \cite{combettes2011proximal}
\begin{align}\label{equ:FB2}
    x^{(n+1)} = prox_{f_1}^{\mu} \left( x^{(n)} - \mu\bigtriangledown f_2(x^{(n)})  \right)
\end{align}
where step size $\mu$ is set within the interval $\left(0, \frac{2}{L} \right)$. In this paper, the function $f_2$ is the data fidelity term and takes the form of $\frac{\|y - \mathcal{A}x\|_2^2}{2\sigma^2}$ from (\ref{equ:miniCauchy}) whilst the function $f_1$ is the Cauchy based penalty function $h$. Following these preliminaries, we can now state the following:

\begin{thm} \label{thm:theorem2}
Let the twice continuously differentiable and non-convex regularisation function $h$ be the function $f_1$ and the $L$-Lipchitz differentiable data fidelity term $\frac{\|y - \mathcal{A}x\|_2^2}{2\sigma^2}$ be the function $f_2$. The iterative FB sub-solution to the optimisation problem in (\ref{equ:miniCauchy}) is
\begin{align}\label{equ:thm2}
    x^{(n+1)} = prox_{Cauchy}^{\mu} \left( x^{(n)} - \frac{\mu\mathcal{A}^T(\mathcal{A}x^{(n)} - y)}{\sigma^2}  \right)
\end{align}
where $\bigtriangledown f_2(x^{(n)}) = \frac{\mathcal{A}^T(\mathcal{A}x^{(n)} - y)}{\sigma^2}$. If the condition
\begin{align}\label{equ:Proofthm2}
    \gamma \geq \frac{\sqrt{\mu}}{2}
\end{align}
holds, then the sub-solution of the FB algorithm is strictly convex, and the FB iteration in (\ref{equ:thm2}) converges to the global minimum.
\end{thm}

\begin{proof}
At each iteration $n$, in order to obtain the iterative estimate $x^{(n+1)}$, by comparing to (\ref{equ:prox}) and (\ref{equ:thm2}), we solve
\begin{align}
    x^{(n+1)} = \arg\min_u G(u)
\end{align}
where the function $G:\mathbb{R}^N \rightarrow \mathbb{R}$ is
\begin{align}
    G(u) = \dfrac{\left\| x^{(n)} - \frac{\mu\mathcal{A}^T(\mathcal{A}x^{(n)} - y)}{\sigma^2} - u \right\|_2^2}{2\mu} - \log\left(\dfrac{\gamma}{\gamma^2 + u^2} \right).
\end{align}

Guaranteeing a convex minimisation problem at each FB iteration will make the whole process convex. As a result, the iterative procedure in (\ref{equ:thm2}) converges to the global minimum of $G$.

Thus, for the cost function $G$ to be convex, the condition $\bigtriangledown^2 G(u) \succeq 0$ should be satisfied. Calculating the Hessian of $G$, we have
\begin{align}\label{equ:thm2GHess}
    \bigtriangledown^2 G(u) = \dfrac{\mathcal{I}}{\mu} + \dfrac{2\gamma^2 - 2u^2}{u^4 + 2\gamma^2u^2 + \gamma^4} \succeq 0,
\end{align}
It is straightforward to show that the required condition to satisfy (\ref{equ:thm2GHess}) can be obtained in the same way as in (\ref{equ:thmFHess}). Hence, the rest of the proof follows that of Lemma \ref{lem:lemma2}.

Consequently, despite having a non-convex penalty function, the FB sub-problem corresponding to the cost function $G$ is strictly convex and converges to the global minimum, with the condition
\begin{align}\label{equ:Proofthm22}
    \gamma \geq \frac{\sqrt{\mu}}{2}.
\end{align}
\end{proof}

\begin{rem}
Note that satisfying the convexity condition for the Cauchy proximal operator via Lemma \ref{lem:lemma2} guarantees the convergence of the general solution via the iterative algorithm~(\ref{equ:thm2}). For this, either the step size $\mu$ can be set based on a $\gamma$ value estimated directly from the observations, or alternatively, $\gamma$ can be set in cases when the Lipchitz constant $L$ is computed and/or estimating $\gamma$ is ill-posed.
\end{rem}

\begin{rem}
Since the data fidelity function $f_2$ is convex and $L$-Lipchitz differentiable, using ADMM or DR algorithms instead of FB in solving the minimisation problem in (\ref{equ:FB1}) for the non-convex Cauchy based penalty function whilst satisfying condition (\ref{equ:Proofthm22}), will not change anything and therefore, their solutions converge to a solution. Thus, the FB approach in Cauchy proximal splitting algorithm given in Algorithm 2 can be replaced with other splitting algorithms.
\end{rem}

\begin{rem}\label{rem:remark6}
The non-convex Cauchy penalty function proposed in this paper guarantees convergence to a minimum by satisfying either (i) $\mathcal{A}^T\mathcal{A} \approx r\mathcal{I}$ (including $r=1$) along with the condition in Theorem \ref{thm:theorem1_new}, or (ii) just the condition from Theorem \ref{thm:theorem2} via a proximal splitting method such as the FB algorithm.
\end{rem}

The convergence-guaranteed Cauchy proximal splitting algorithm is given in Algorithm \ref{alg:FB}.
\begin{algorithm}[ht!]
\caption{Cauchy proximal splitting (CPS) algorithm}\label{alg:FB}
\setstretch{1.2}
\begin{algorithmic}[1]
\State \textbf{Input:} $\text{Data, }y \text{ and } MaxIter$
\State \textbf{Input:} $\mu\in\left(0, \frac{2}{L}\right) \text{ and } \gamma\geq\frac{\sqrt{\mu}}{2}$
\State \textbf{Set:} $i\gets0 \text{ and } x^{(0)}$
  \Do
    \State $u^{(i)} \gets x^{(i)} - \mu \mathcal{A}^T(\mathcal{A}x^{(i)} - y)$
    \State $x^{(i+1)} \gets \textsc{proxCauchy}(u^{(i)}, \gamma, \mu) \text{ via Algorithm \ref{alg:proxCauchy}}$
    \State $i++$
  \doWhile{$\dfrac{\|x^{(i)} - x^{(i-1)}\|}{\|x^{(i-1)}\|} > \varepsilon \text{ or } i<MaxIter$}
\end{algorithmic}
\end{algorithm}

\section{Experimental Analysis}\label{sec:results}
We focus the experimental part of this paper on three separate applications. First, we evaluate the proposed approach on 1D signal denoising in the frequency domain. Secondly, we investigate it when applied to two classical image processing tasks, i.e. denoising and de-blurring. Finally, we have illustrated an application to error recovery in MIMO communication systems.

\subsection{Signal Denoising in Frequency Domain}
The first example demonstrates the use of the non-convex Cauchy based penalty function in 1D signal denoising application. In particular, we consider the classical sinusoidal signal  "Heavy Sine" which is included in Matlab distributions. This signal was analysed in additive white Gaussian noise (AWGN) of several levels, with signal-to-noise-ratio (SNR) values between 2 and 12 decibels (dB).

We synthesised the signal $y \in \mathbb{R}^M$ via an over-sampled discrete inverse Fourier transform operator $\mathcal{F}^{-1}$ as $y = \mathcal{F}^{-1}x + n$, where $x\in\mathbb{C}^{N}$. The number of points in the frequency domain was chosen as $N = 512, 2048,$ and $8192$ whereas the number of signal samples was $M = 128, 256, 512$, respectively. We created three simulation cases corresponding to compression ratios ($M/N$) of 0.25, 0.125 and 0.0625. The operator $\mathcal{F}$ is a normalised tight frame with $\mathcal{F}^H\mathcal{F}=\mathcal{I}$. We compared the performance of the Cauchy based penalty function with $L_1$ and $TV$ norm penalty functions. The root-mean square error (RMSE) and mean absolute error (MAE) were used as evaluation metrics in this case.

The first experiment is depicted in Figure \ref{fig:gammaComp1D}, which shows the effect of the scale parameter $\gamma$ on denoising results both when violating and when satisfying the conditions proposed for convergence. Specifically, the vertical red and black dotted-lines show the scale parameter value for $\gamma = \sigma/2\sqrt{r}$ from Theorem \ref{thm:theorem1_new} and $\gamma = \sqrt{\mu}/2$ from Theorem \ref{thm:theorem2}, respectively.  A range of values for $\gamma$ between $10^{-2}$ and $10^{2}$ was set, and denoised signals were obtained for each $\gamma$ values by using the Algorithm \ref{alg:FB}. The error term $\varepsilon$ was set to $10^{-3}$ whilst the maximum number of iterations $MaxIter$ was set to 500. We follow \cite{combettes2011proximal} for the selection of the step size $\mu$ and then use Theorems \ref{thm:theorem1_new} and \ref{thm:theorem2} to decide the minimum value for $\gamma$ that guarantees convergence. From the definition \cite{combettes2011proximal}, the data fidelity term $\|y - \mathcal{F}^{-1}x\|_2^2$ is convex and differentiable with a $L$-Lipschitz continuous gradient, where $L$ is the Lipschitz constant. Thus, we can select the step size $\mu$ within the range $\left(0, \frac{2}{L}\right)$. There is no strict rule in choosing the $\mu$ values, but the literature suggests that choosing $\mu$ close to $\frac{2}{L}$ is more efficient. Hence, for this example, we decided to set $\mu=\frac{3}{2L}$.
On examining Figure \ref{fig:gammaComp1D}, it is clear that the lowest RMSE value is achieved for a $\gamma$ value higher than the critical values shown with red and black doted-lines. It can also be seen that $\gamma$ values 2-3 times higher than both critical values give relatively good results when compared to those with $\gamma$ values which are 20 times higher.

\begin{figure}[t!]
\centering
\includegraphics[width=.6\linewidth]{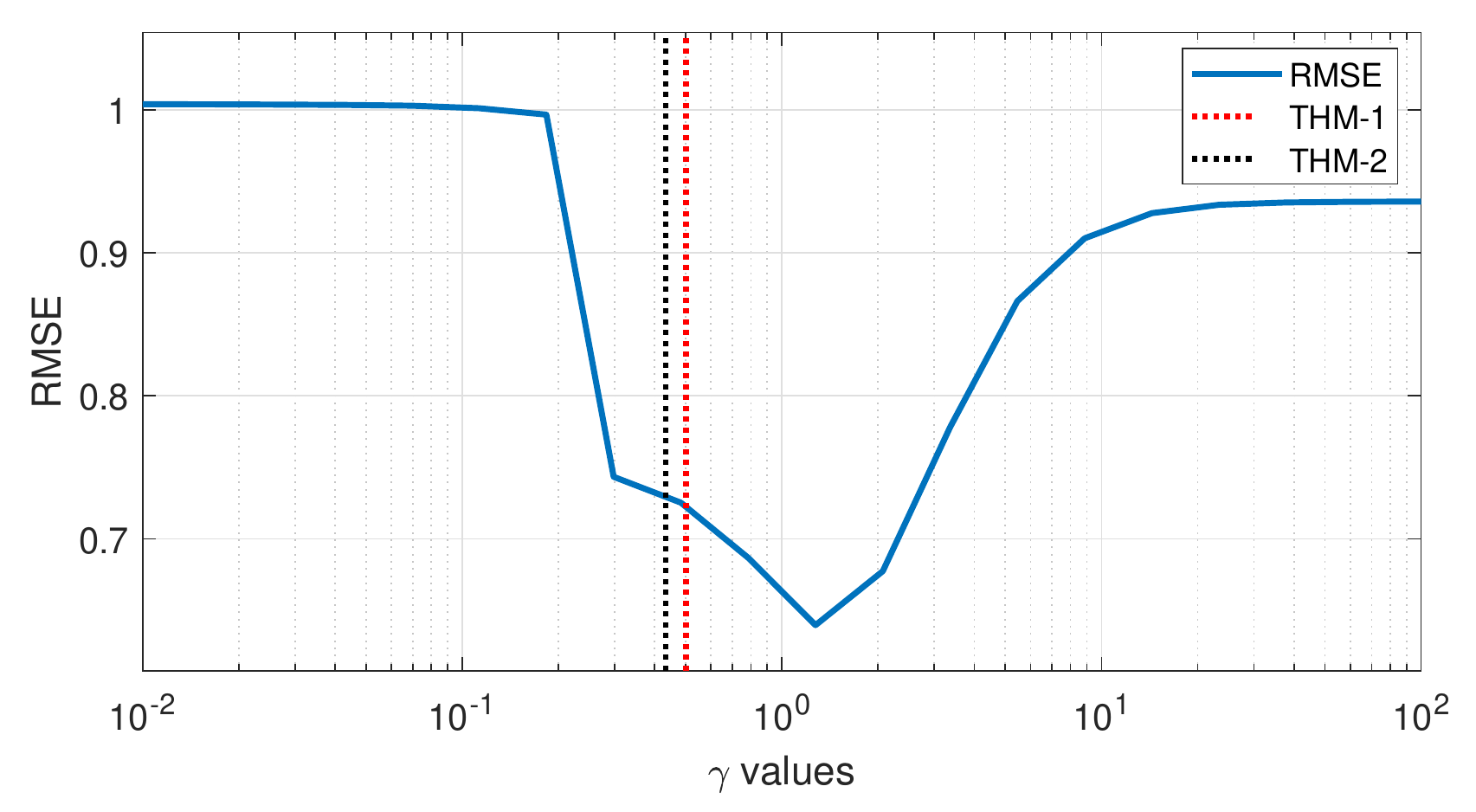}
\caption{Effect of the scale parameter $\gamma$ on the denoised signals ($M/N = 0.25$).}
\label{fig:gammaComp1D}
\end{figure}

\begin{table}[htbp]
  \centering
  \caption{1D Denoising performance for different $M/N$ values.}
    \resizebox{0.65\linewidth}{!}{\begin{tabular}{cl|rrr|rrr}
    \toprule
          &       & \multicolumn{3}{|c|}{SNR = 4dB} & \multicolumn{3}{c}{SNR = 10dB} \\
          & Performance & \multicolumn{3}{|c|}{Methods} & \multicolumn{3}{c}{Methods} \\
    $M/N$   & Metric & \multicolumn{1}{|l}{$L_1$} & \multicolumn{1}{l}{$TV$} & \multicolumn{1}{l|}{Cauchy} & \multicolumn{1}{l}{$L_1$} & \multicolumn{1}{l}{$TV$} & \multicolumn{1}{l}{Cauchy} \\
    \toprule
    \multirow{2}[0]{*}{0.25} & RMSE  & 0.6481 & 0.6821 & \textbf{0.5687} & 0.9375 & 0.8405 & \textbf{0.8156} \\
          & MAE   & 0.4697 & 0.5266 & \textbf{0.3603} & 0.6135 & 0.5949 & \textbf{0.4952} \\
          \hline
    \multirow{2}[0]{*}{0.125} & RMSE  & 0.6418 & 0.6799 & \textbf{0.4649} & 0.9405 & 0.7676 & \textbf{0.6953} \\
          & MAE   & 0.4697 & 0.5266 & \textbf{0.3603} & 0.6135 & 0.5949 & \textbf{0.4952} \\
          \hline
    \multirow{2}[0]{*}{0.0625} & RMSE  & 0.6392 & 0.7305 & \textbf{0.4585} & 0.9324 & 0.8104 & \textbf{0.6974} \\
          & MAE   & 0.4600 & 0.5645 & \textbf{0.3492} & 0.5859 & 0.6365 & \textbf{0.4718} \\
          \bottomrule
    \end{tabular}}%
  \label{tab:1D}%
\end{table}%

\begin{figure}[ht!]
\centering
\includegraphics[width=.6\linewidth]{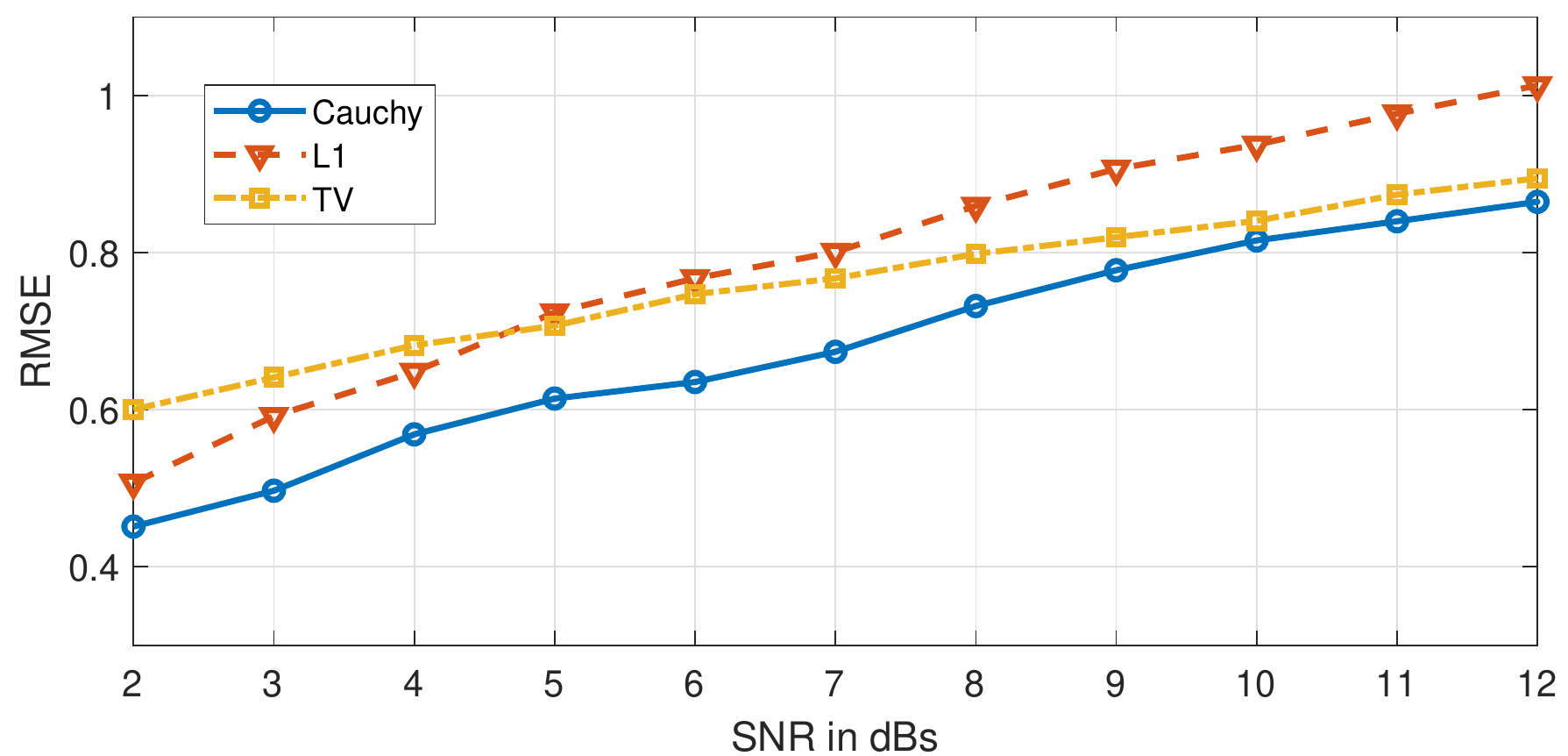}\\
\includegraphics[width=.6\linewidth]{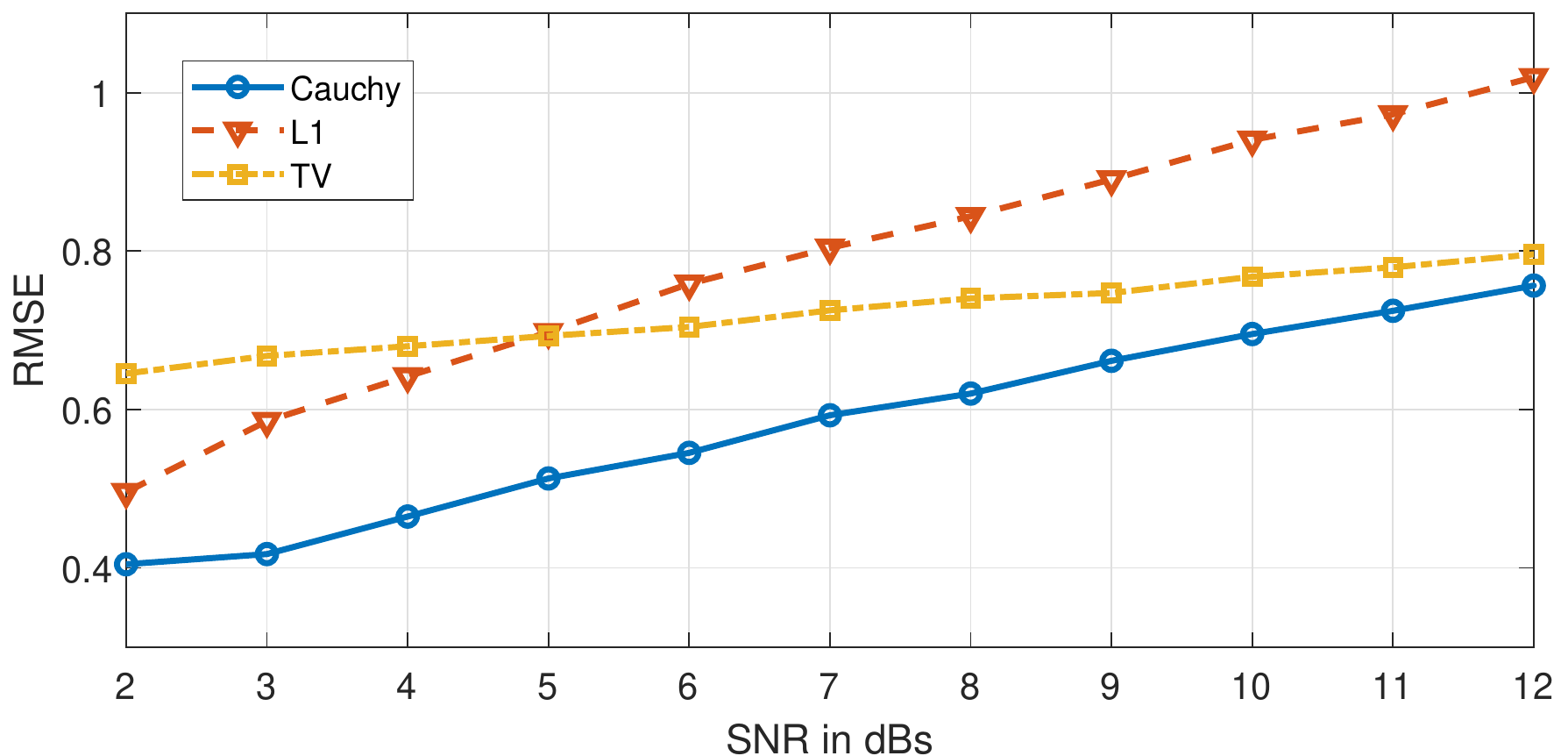}\\
\includegraphics[width=.6\linewidth]{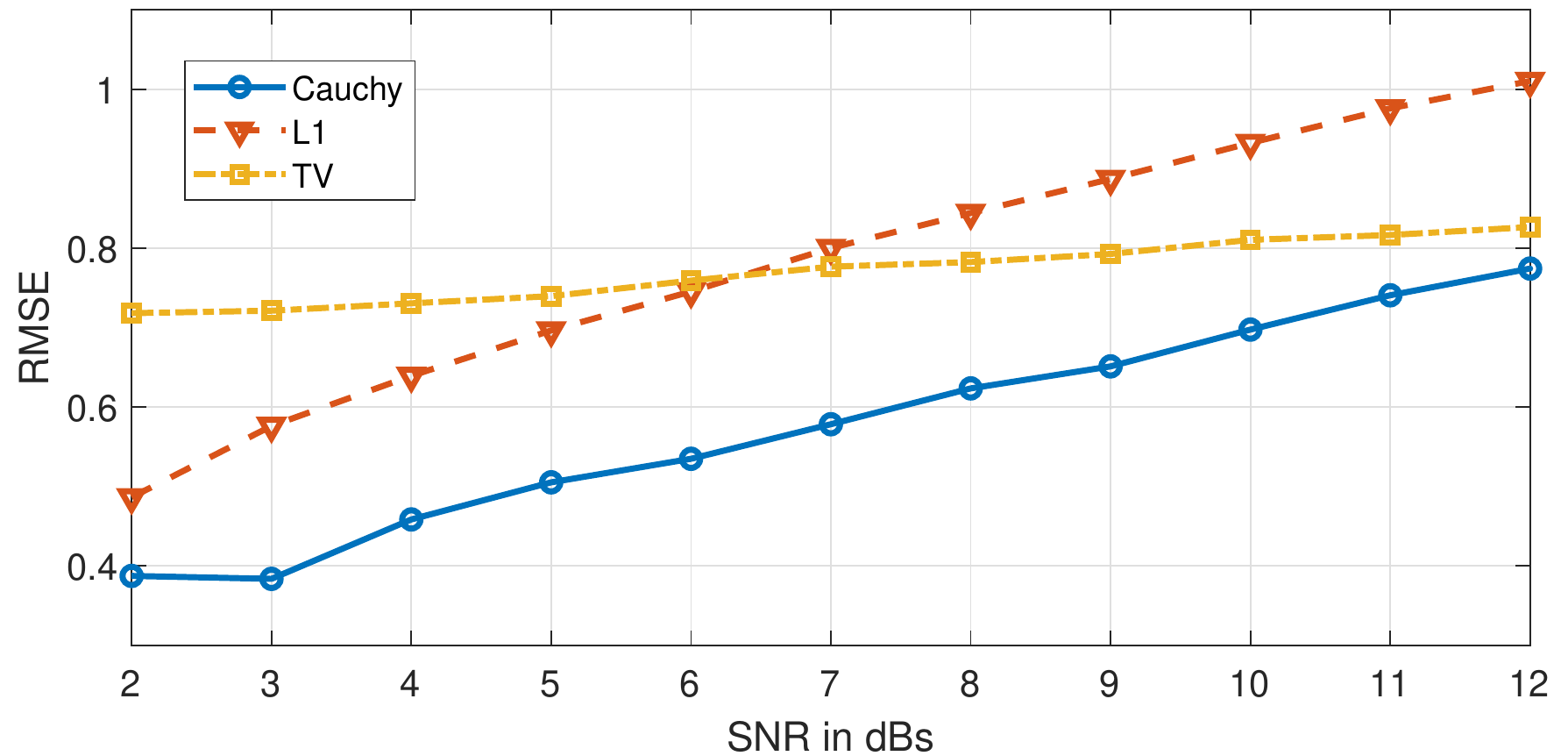}
\caption{Denoising performance in terms of RMSE for various penalty functions. From top to bottom, sub-figures show plots for compression ratios of 0.25 (M = 128, N = 512), 0.125 (M = 256, N = 2048), and 0.0625 (M = 512, N = 8096), respectively.}
\label{fig:rmseComp1D}
\end{figure}

In order to further assess the performance of the proposed Cauchy denoiser, we calculated RMSE and MAE values corresponding to initial SNR values between 2 and 12 dBs. For each noise level, simulations were repeated 100 times and corresponding average RMSE and MAE values for each penalty function, for two example SNR values, are presented in Table \ref{tab:1D}. Figure \ref{fig:rmseComp1D} depicts the performance results in terms of RMSE for the three compression ratio values considered. It can be seen from Table \ref{tab:1D} that the lowest RMSE and MAE values are obtained when employing the Cauchy based penalty function for all compression ratio values. When evaluating Figure \ref{fig:rmseComp1D}, it can be seen that TV denoising performance is close to that of the proposed penalty function when decreasing the noise level, in all three cases. Nevertheless, the Cauchy penalty function leads to the best performance for all SNR values. For visual assessment, Fig. \ref{fig:1D_figures} shows denoising results corresponding to $L_1$, $TV$ and Cauchy based penalty functions for an SNR of 7 dBs. For all the penalty functions tested the denoising effect can be clearly seen but the proposed penalty function leads to the lowest RMSE and MAE.

\begin{figure}[ht!]
\centering
\includegraphics[width=0.6\linewidth]{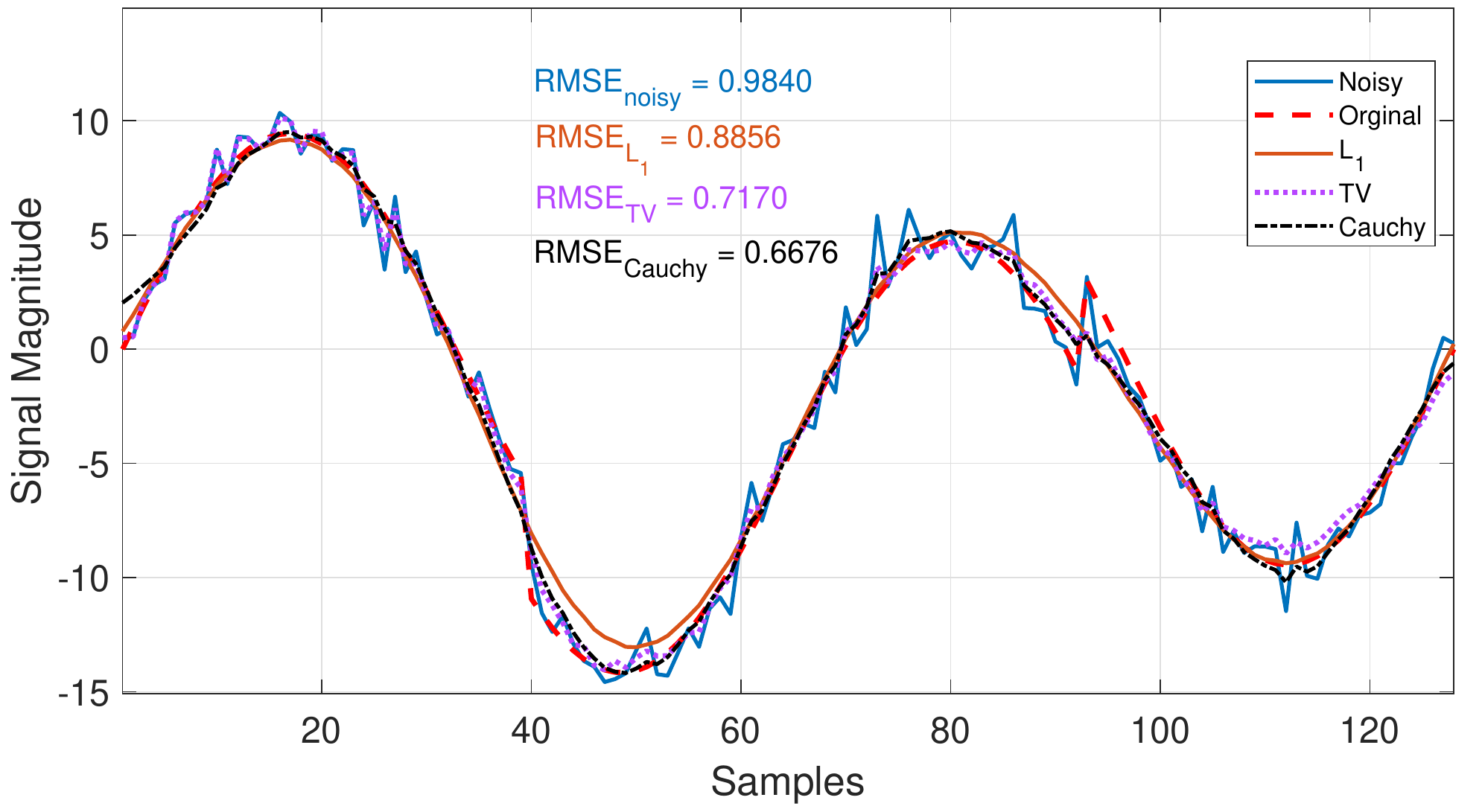}
\caption{Denoising using various penalty functions for a signal of length M = 128 where N is 512 (compression ratio 0.25).}
\label{fig:1D_figures}
\end{figure}

\subsection{2D Image Reconstruction}
In the second set of experiments, we investigated the influence of the proposed Cauchy-based regularisation on the
classical 2D image reconstruction tasks of denoising and deblurring. Both problems are ill-posed due to the nature of the measurement noise and unknown blurring-kernels. The literature includes various approaches which can be classified into two main groups, including optimisation-based \cite{parekh2015convex,li2019blind,danielyan2011bm3d} and learning-based \cite{7274732,sun2015learning} techniques. In optimisation based approaches, the sparse characteristics of the object of interest, make sparsity enforcing penalty functions play a leading role.

We consider for this set of experiments three types of images, including the standard Cameraman, magnetic resonance (MR) and synthetic aperture radar (SAR). For deblurring, the forward operator $\mathcal{A}$ was selected as a 5$\times$5 Gaussian point spread function (PSF) with standard deviation of 1. The noise is AWGN with blurred-signal to noise ratio (BSNR = $10\log_{10}\{var(\mathcal{A}x)/\sigma^2\}$ (where $\sigma^2$ is the Gaussian noise variance) of 40 dBs. For the denoising example, the forward operator $\mathcal{A}$ is the identity matrix $\mathcal{I}$, the additive noise corresponds to an SNR of 20 decibels.

We start by discussing the effects of the scale parameter, $\gamma$ on the reconstruction results depending on whether the conditions in Theorems 1 and 2 are violated or satisfied. We used the standard cameraman image for benchmarking. The analysis was performed in terms of the peak signal to noise ratio (PSNR) and RMSE. A range of values for $\gamma$ between $10^{-4}$ and $10^{4}$ was set, and the reconstructed images were obtained for each $\gamma$ values by using Algorithm \ref{alg:FB}. The error term $\varepsilon$ was set to $10^{-3}$ whilst the maximum number of iterations $MaxIter$ was set to 250. The step size $\mu$ was set to $\frac{3}{2L}$ for this example.

\begin{figure}[ht!]
\centering
\subfigure[]{\includegraphics[width=.6\linewidth]{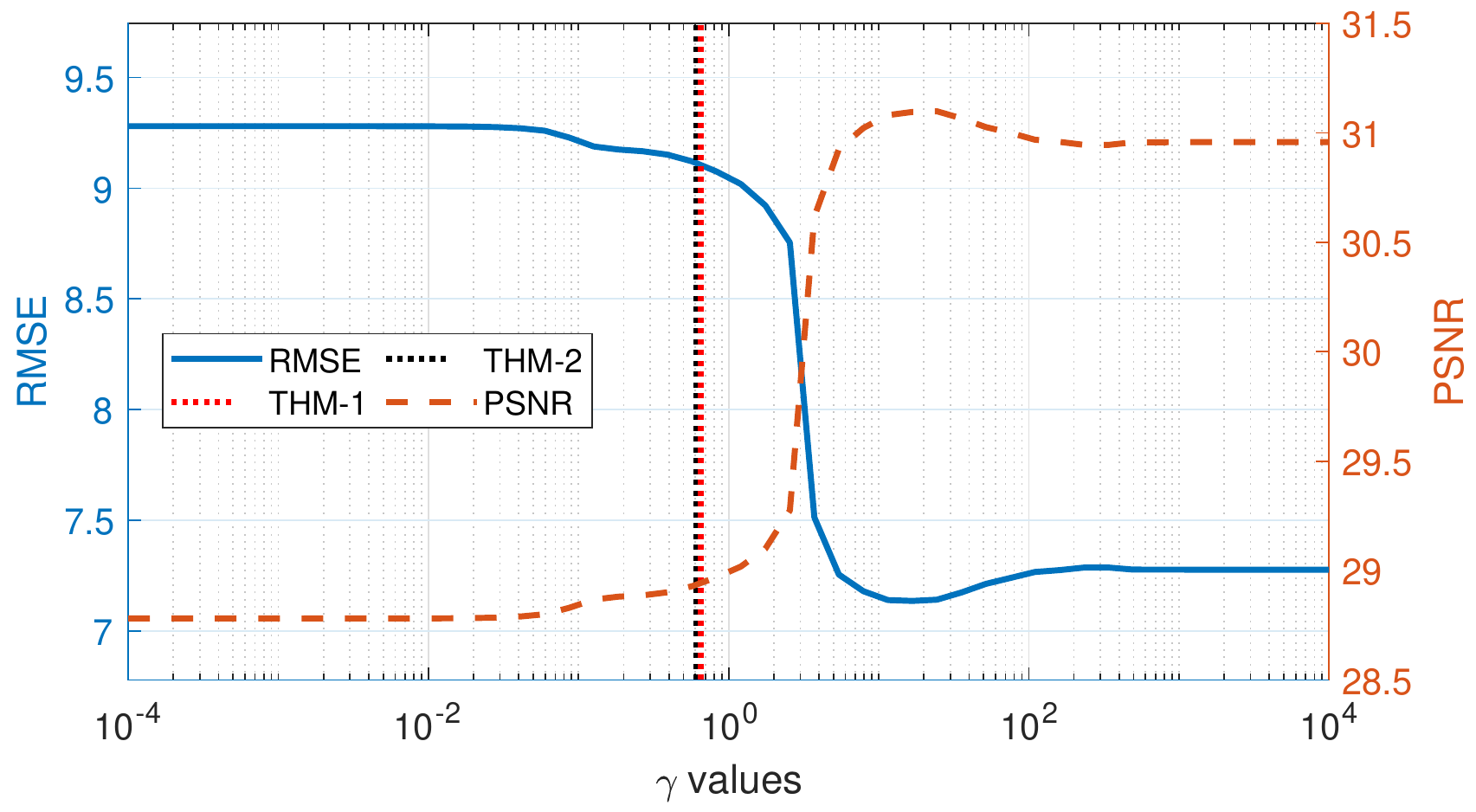}}
\subfigure[]{\includegraphics[width=.6\linewidth]{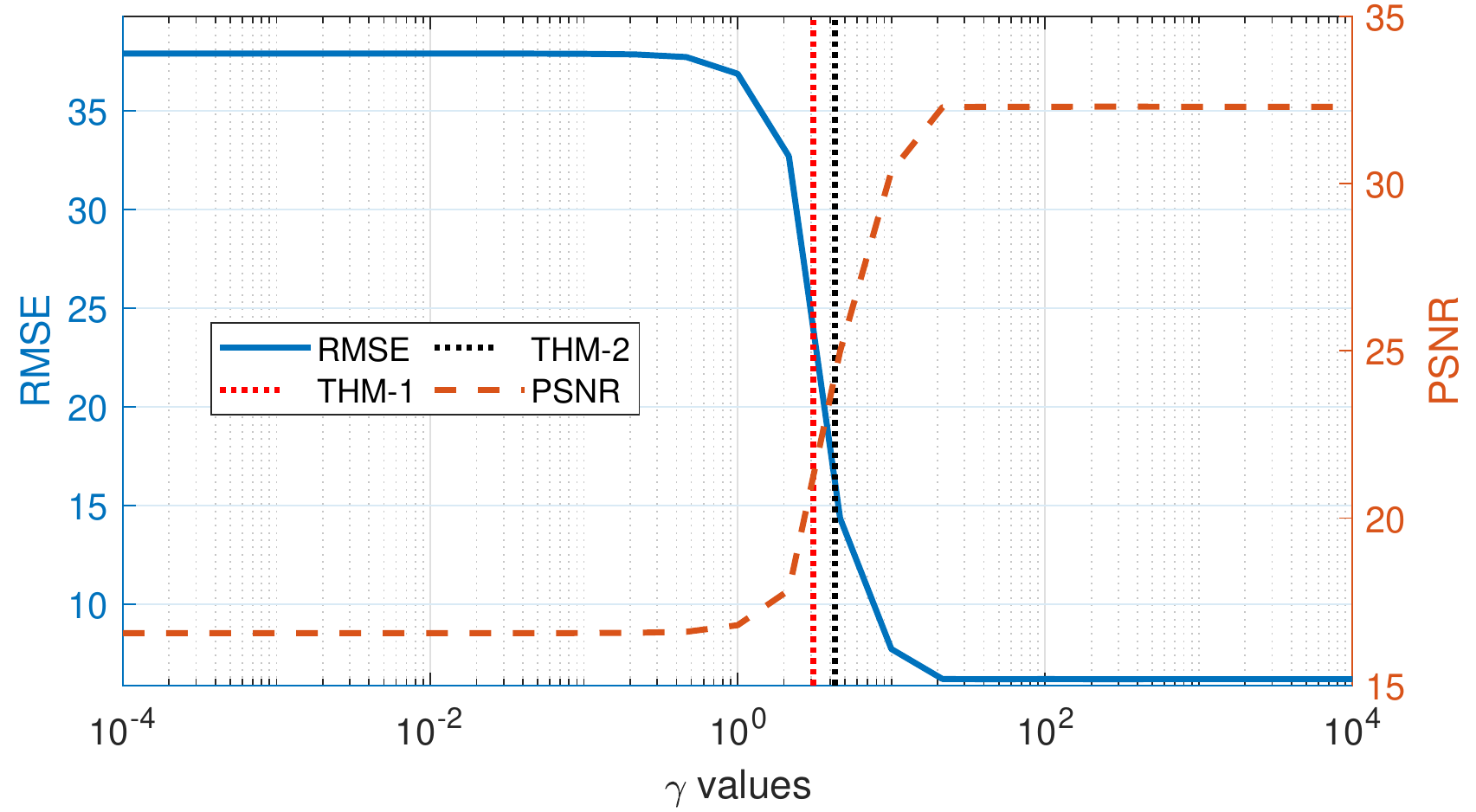}}
\caption{Effect of the scale parameter $\gamma$ on the reconstructed images for Cameraman. (a) deblurring, (b) denoising.}
\label{fig:gammaComp}
\end{figure}

Figure \ref{fig:gammaComp} shows the effect of $\gamma$ values on reconstruction results. The left $y$-axes in both sub-figures show RMSE values whilst the right $y$-axes represent the PSNR values for different vaues of $\gamma$ on the $x$-axes. As can clearly be seen from both sub-figures, reconstruction results are poor when the conditions in both Theorem \ref{thm:theorem1_new} and \ref{thm:theorem2} (left sides of the vertical dotted-lines) are violated. However, starting from either conditions and higher values of $\gamma$, we obtained better reconstruction results with an important reconstruction gain around 16dBs for denoising and 2 dBs for deblurring in terms of PSNR. This proves experimentally the correctness of the conditions derived in Theorems \ref{thm:theorem1_new} and \ref{thm:theorem2}. Unlike in the the 1D case, for image reconstruction, we observe a similar performance for higher values of $\gamma$ .
We conclude that there is no strict rule for choosing the optimum  value of $\gamma$ but we noticed that the best performance is generally achieved within a specific interval and hence we recommend using $\gamma \in \left[\frac{\sqrt{\mu}}{2}, \frac{50\sqrt{\mu}}{2} \right]$.


\begin{figure*}[htbp]
\centering
\subfigure[]{\includegraphics[width=.24\linewidth]{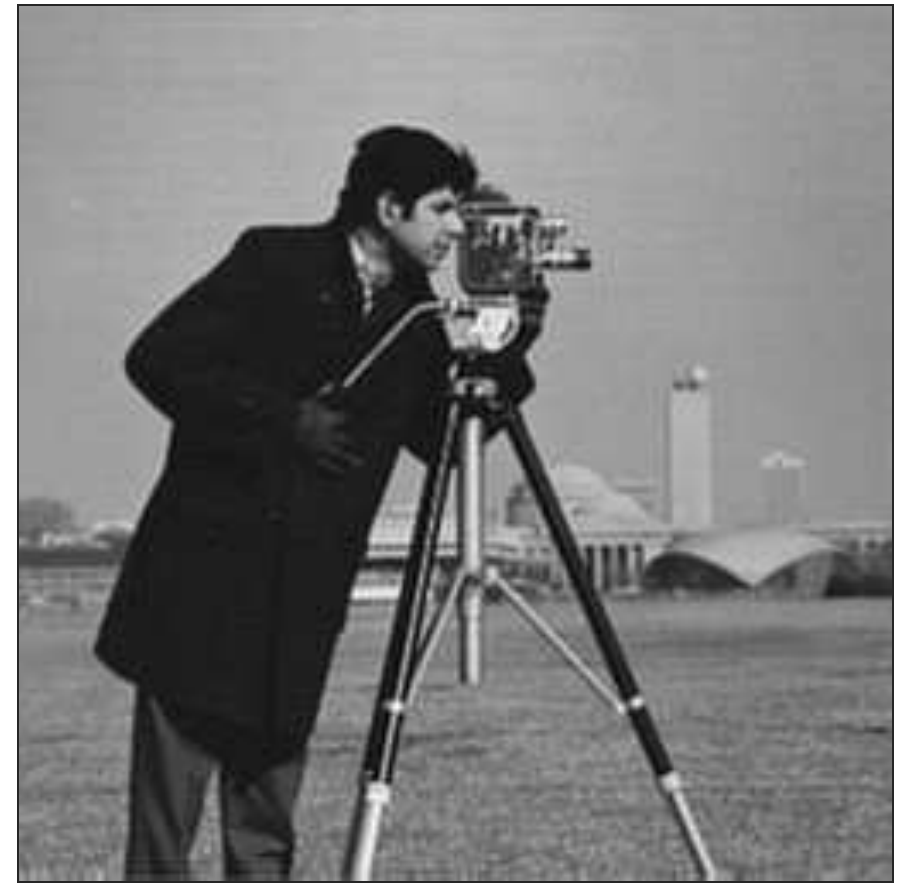}}
\centering
\subfigure[]{\includegraphics[width=.24\linewidth]{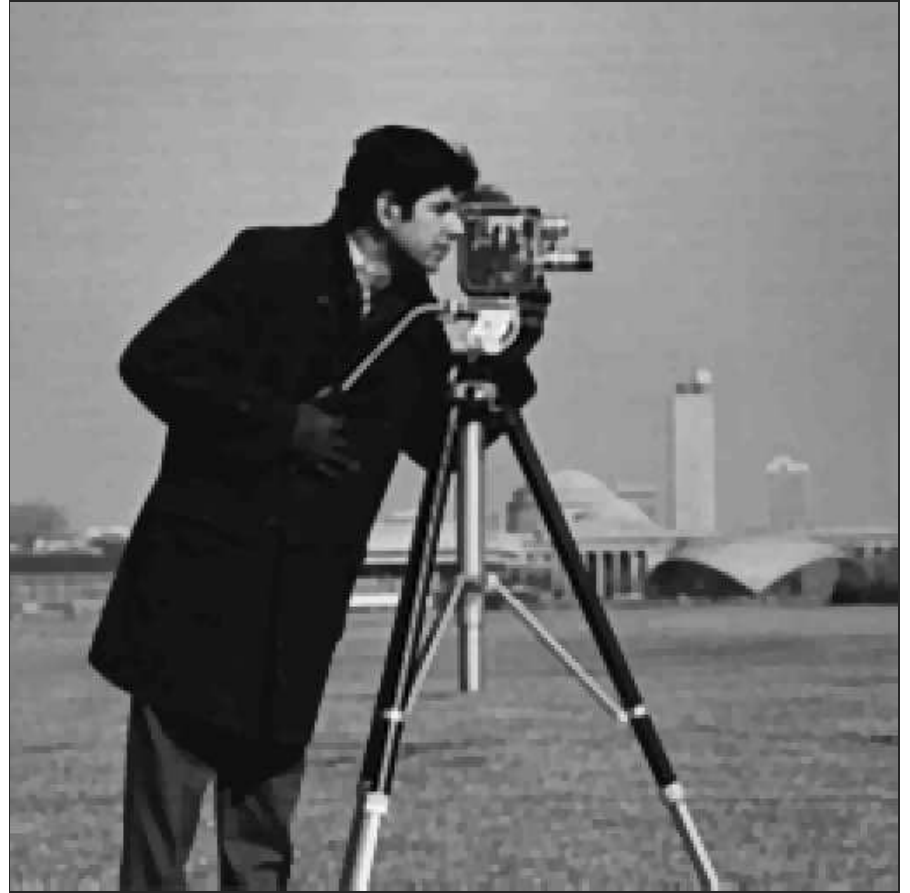}}
\centering
\subfigure[]{\includegraphics[width=.24\linewidth]{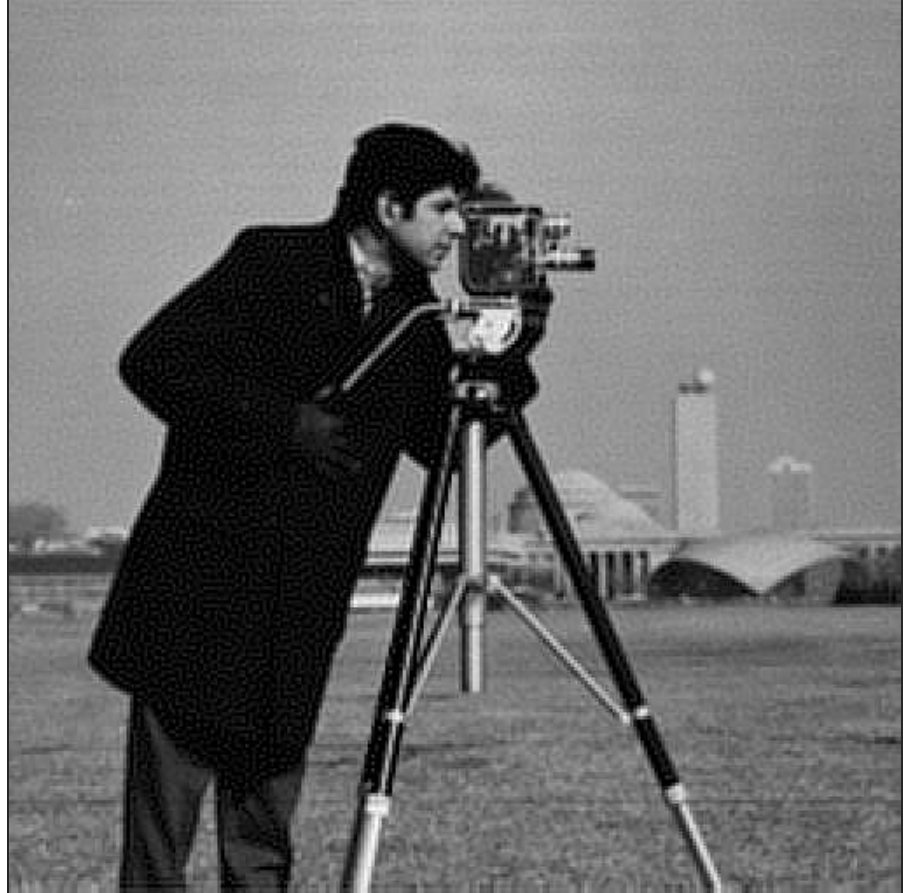}}
\centering
\subfigure[]{\includegraphics[width=.24\linewidth]{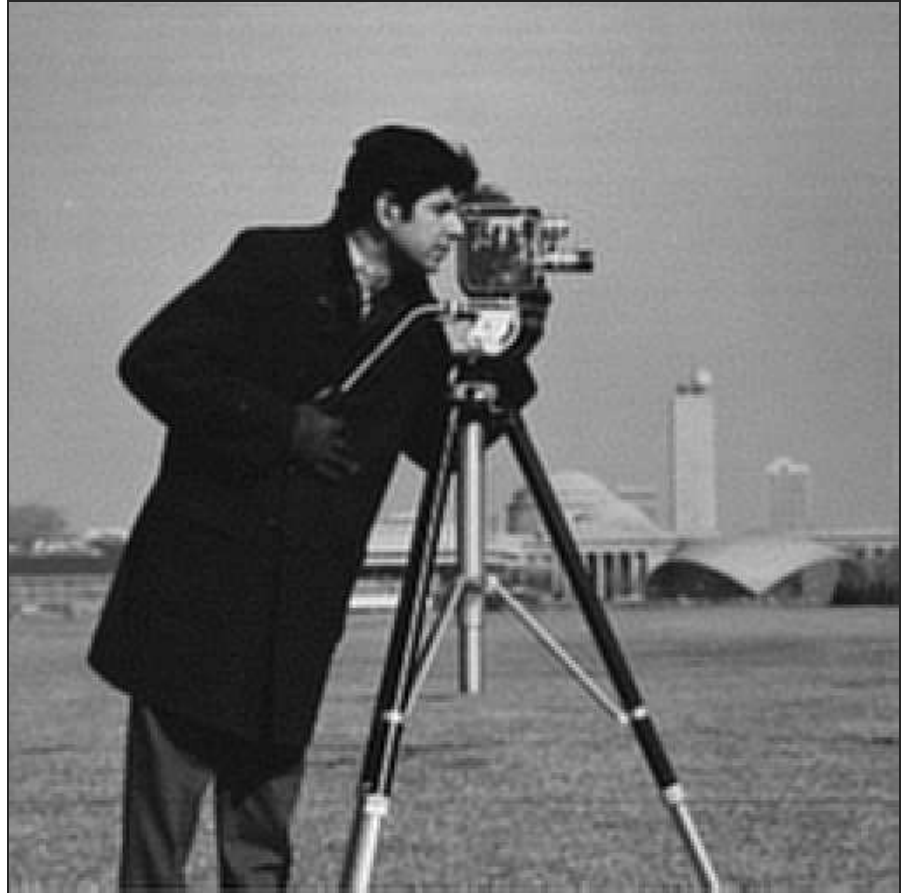}}

\centering
\subfigure[]{\includegraphics[width=.24\linewidth]{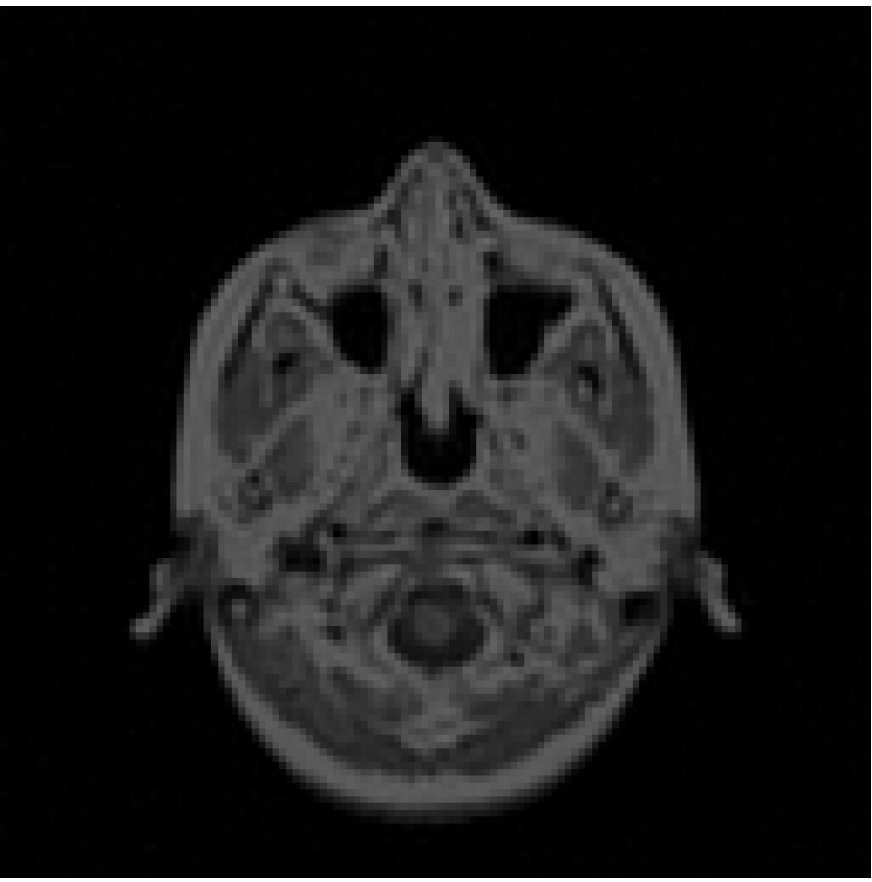}}
\centering
\subfigure[]{\includegraphics[width=.24\linewidth]{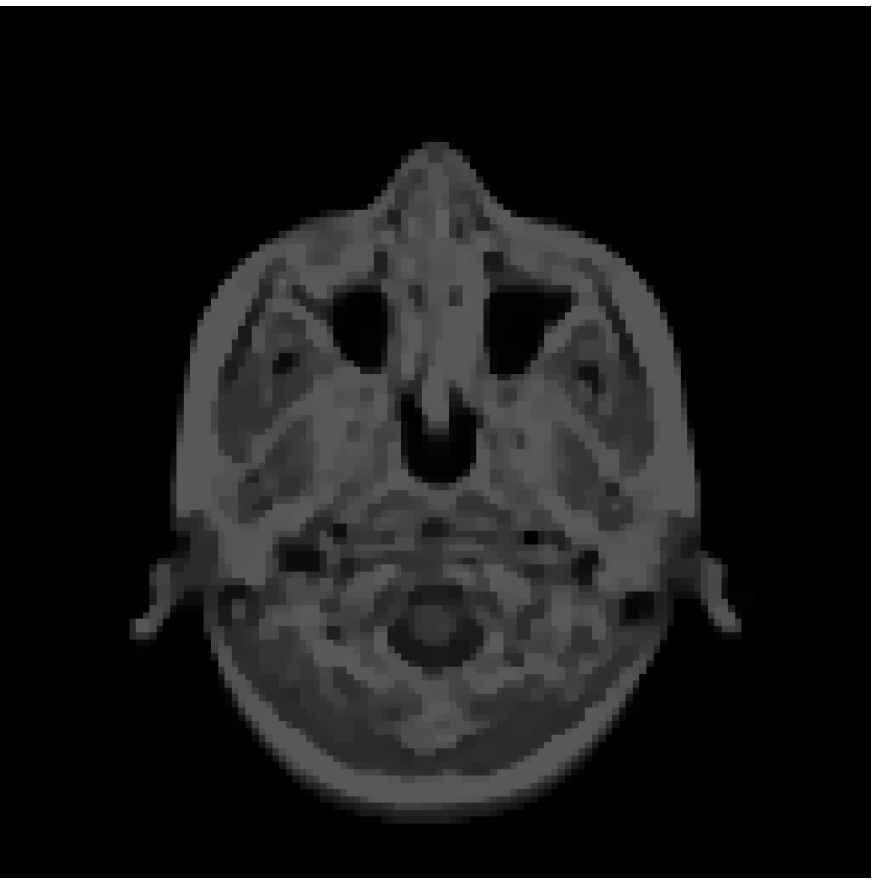}}
\centering
\subfigure[]{\includegraphics[width=.24\linewidth]{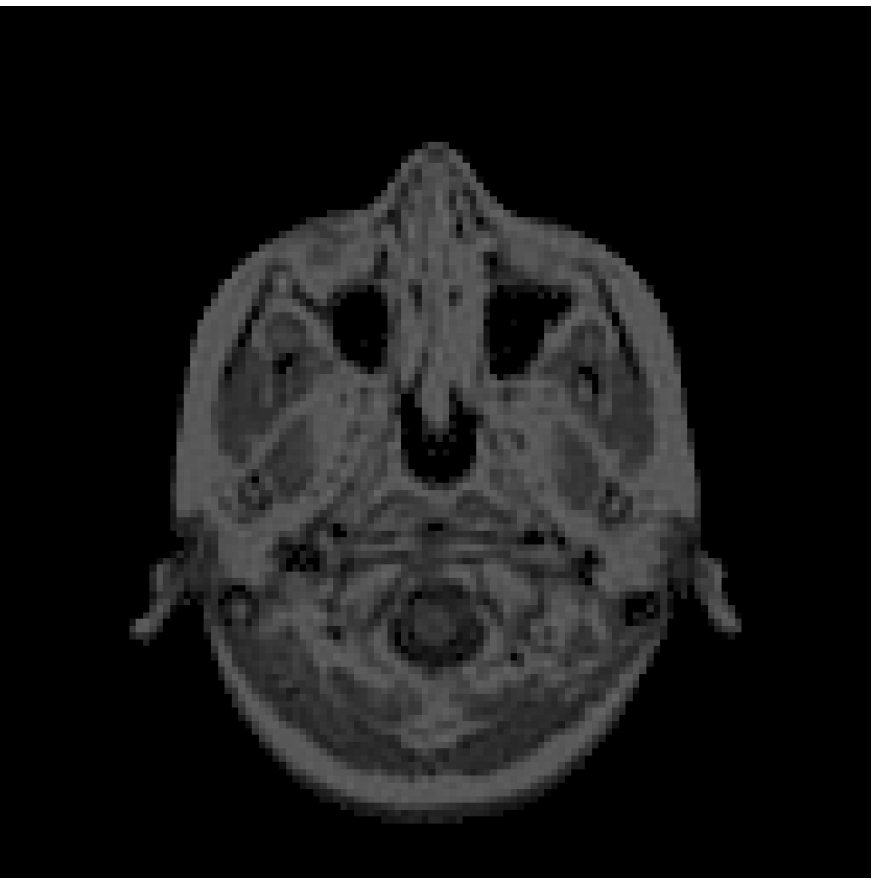}}
\centering
\subfigure[]{\includegraphics[width=.24\linewidth]{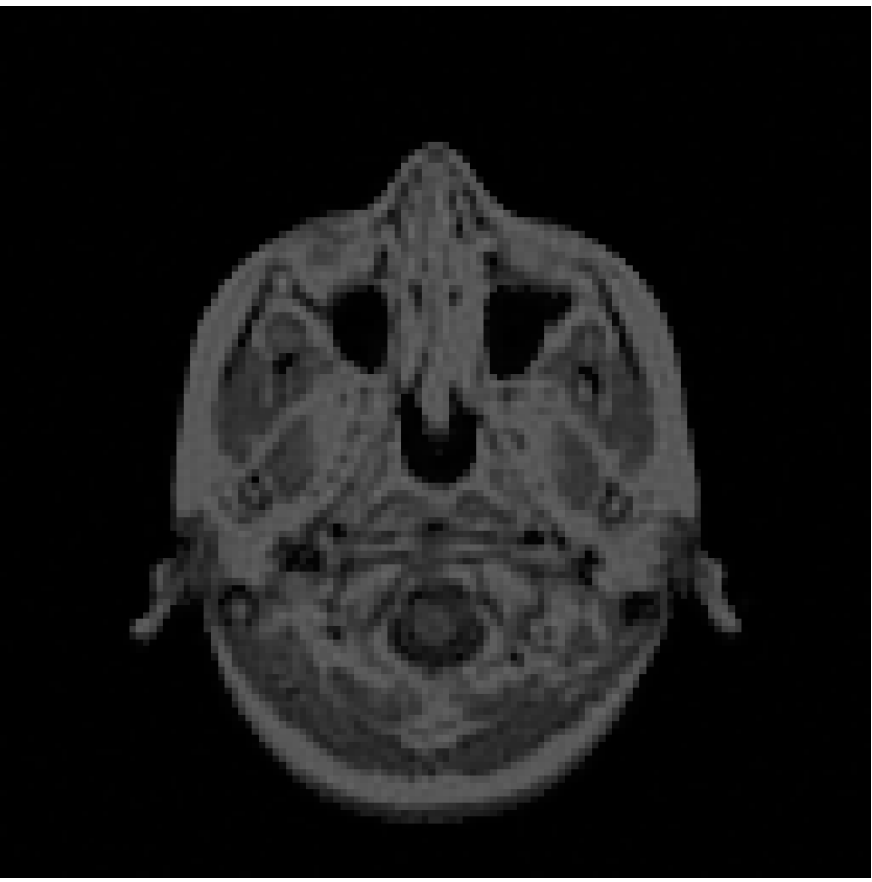}}

\centering
\subfigure[]{\includegraphics[width=.24\linewidth]{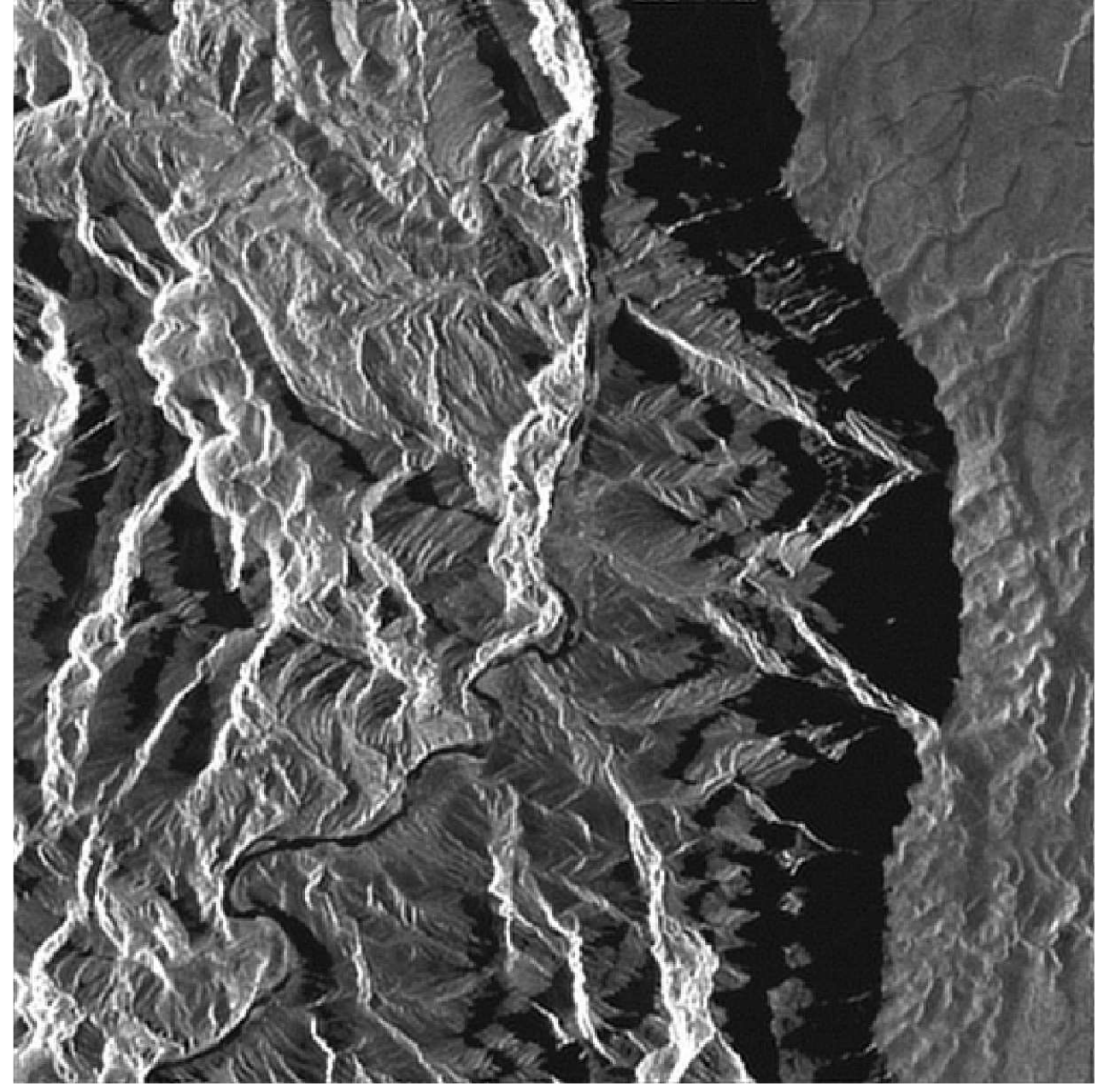}}
\centering
\subfigure[]{\includegraphics[width=.24\linewidth]{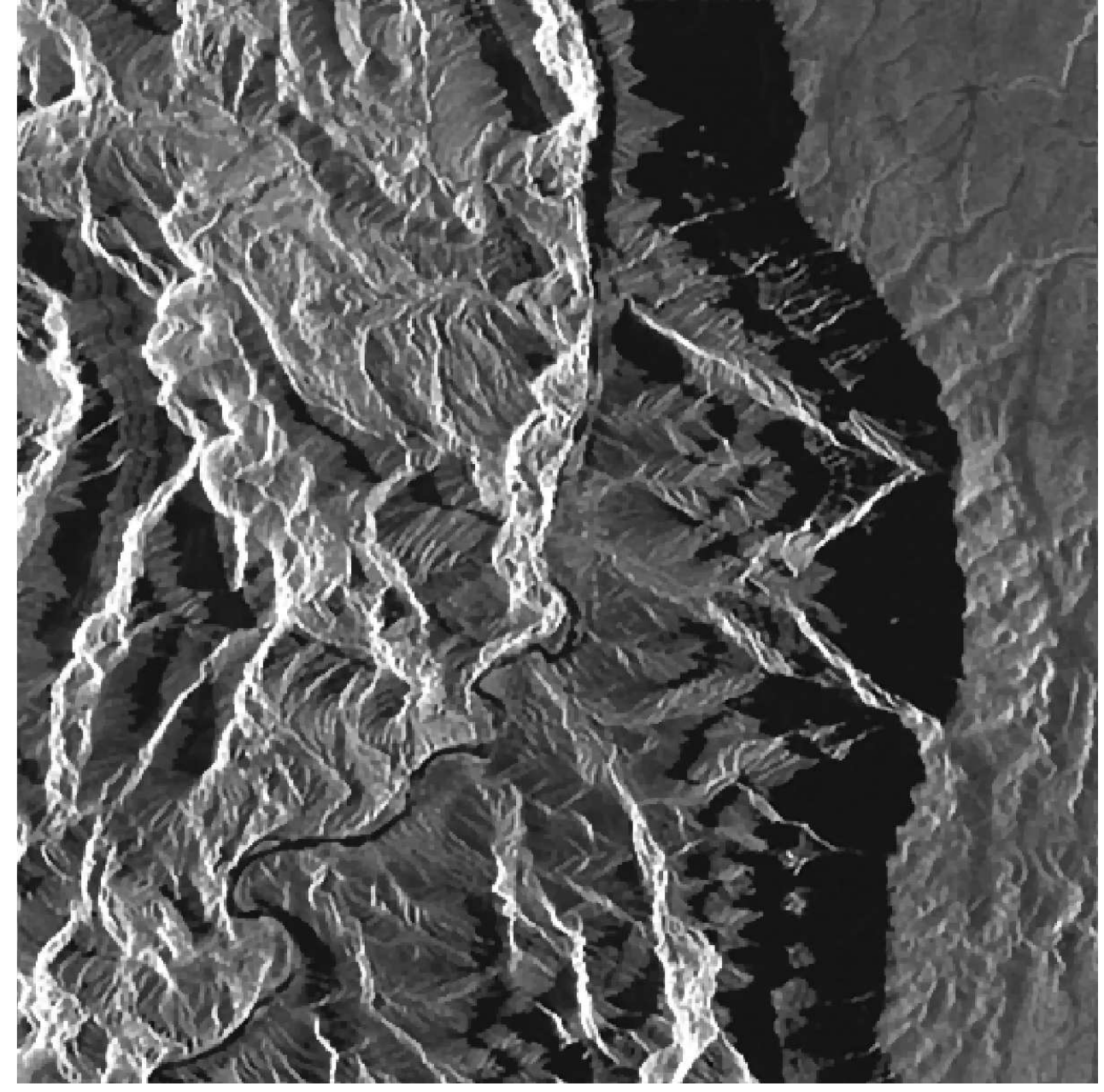}}
\centering
\subfigure[]{\includegraphics[width=.24\linewidth]{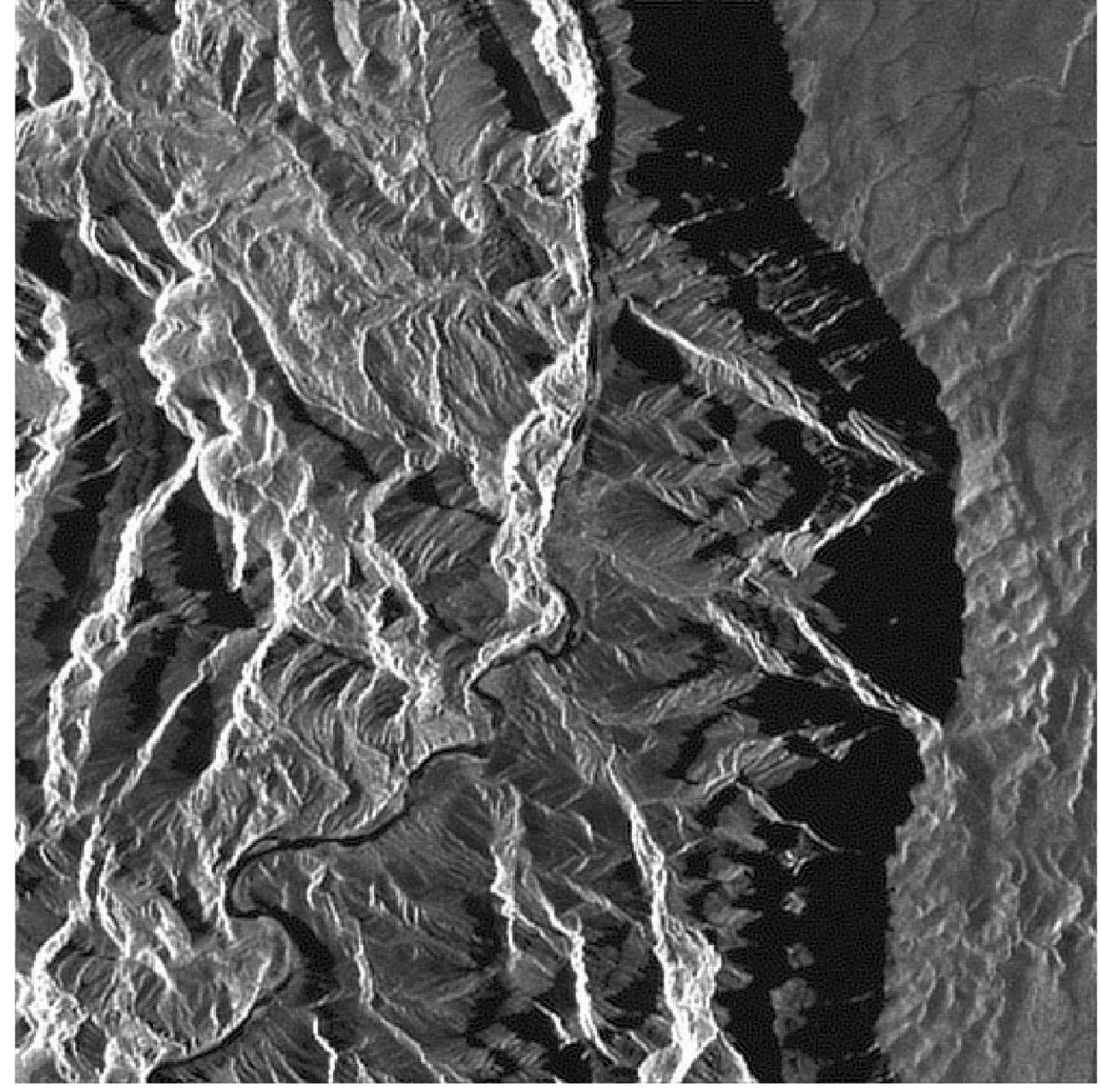}}
\centering
\subfigure[]{\includegraphics[width=.24\linewidth]{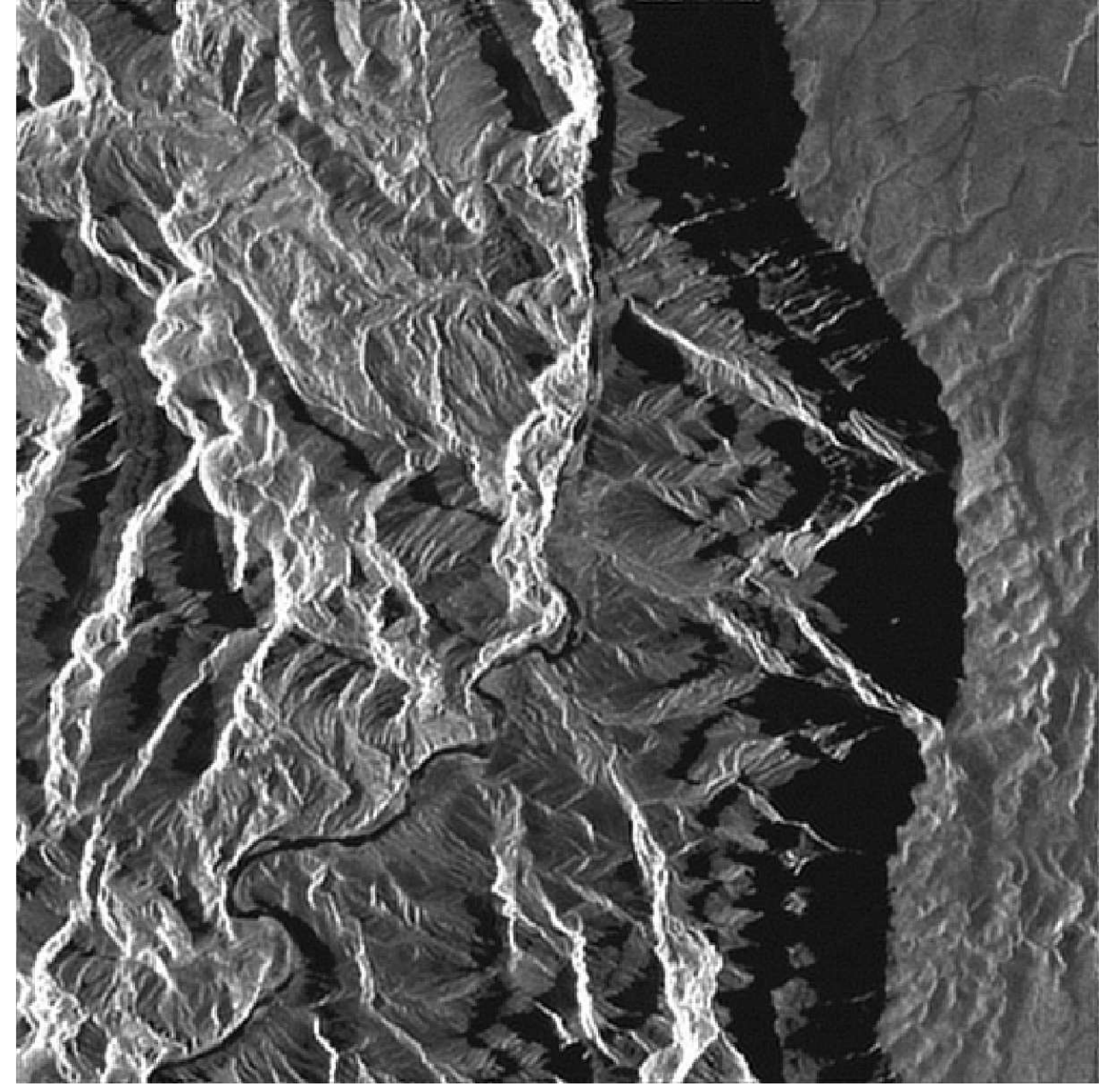}}

\caption{De-blurring results. 
First column results ((a), (e) and (i)) refer to L1 norm penalty with a regularisation constant of 0.01. Second column results ((b), (f) and (j)) refer to TV norm penalty with a regularisation constant of 0.1. Third column results ((c), (g) and (k)) refer to Cauchy with $\gamma = \frac{\sqrt{\mu}}{2}$. Fourth column results ((d), (h) and (l)) refer to Cauchy with $\gamma_{opt}$.}
\label{fig:deblur}
\end{figure*}

Please also note that we do not compare the two conditions proposed in Theorems \ref{thm:theorem1_new} and \ref{thm:theorem2}. They are not antagonistic, but rather conditions that together provide solutions in various situations. Their usage depends on the problem at hand (cf. Remark \ref{rem:remark6}), and both guarantee the convergence in specific circumstances.

\begin{table}[htbp]
  \centering
  \caption{Image de-blurring performance for various images.}
    \begin{tabular}{rlrrr}
    \toprule
          &       & \multicolumn{1}{l}{L1} & \multicolumn{1}{l}{TV} & \multicolumn{1}{l}{Cauchy} \\
          \toprule
    \multicolumn{1}{l}{Cameraman} & PSNR  & 31.07 & 31.03 & \textbf{31.23} \\
          & RMSE  & 7.173 & 7.167 & \textbf{7.035} \\
          & SSIM  & 0.9016 & \textbf{0.9203} & 0.9025 \\
          \hline
    \multicolumn{1}{l}{MRI} & PSNR  & 34.46 & 33.75 & \textbf{35.63} \\
          & RMSE  & 4.875 & 5.240 & \textbf{4.217} \\
          & SSIM  & 0.9547 & 0.9537 & \textbf{0.9707} \\
          \hline
    \multicolumn{1}{l}{SAR} & PSNR  & \textbf{25.78} & 25.19 & \textbf{25.78} \\
          & RMSE  & \textbf{13.174} & 14.026 & \textbf{13.174} \\
          & SSIM  & 0.8897 & 0.8627 & \textbf{0.8898} \\
          \bottomrule
    \end{tabular}%
  \label{tab:deblur}%
\end{table}%

Figure \ref{fig:deblur} depicts de-blurring results for Cauchy, $TV$ and $L1$ norm penalty functions for all three images considered in this set of experiments. For the Cauchy-based reconstruction, we show two separate results corresponding to values of $\gamma$ of $\frac{\sqrt{\mu}}{2}$ (from the convergence condition) and $\gamma_{opt}$, which is the $\gamma$ value with the best performance within the interval of $ \left[\frac{\sqrt{\mu}}{2}, \frac{50\sqrt{\mu}}{2} \right]$. From Figure \ref{fig:deblur} (c) and (k) it can be observed that the Cauchy based penalty function leads to a poor de-blurring performance for the cameraman and SAR images when $\gamma = \frac{\sqrt{\mu}}{2}$. $TV$, $L_1$ and Cauchy-based results are visually similar, but the Cauchy penalty determines the best performance metrics, as shown in Table \ref{tab:deblur}.


\subsection{Error Recovery for MIMO Signal Detection}

For the final part of the experimental analysis in this paper, we investigated an application to error recovery in MIMO signal detection. Due to the sparse nature of error signals in MIMO systems, error recovery has been an important step in estimating the transmitted signals. The state-of-the-art signal detection approaches are the linear zero forcing (ZF) and the minimum mean square error (MMSE) estimators. Variational approaches offer an attractive solution to this problem by promoting the sparse and discrete structure of the error signals. In \cite{7840470}, a maximum-a-posteriori based approach has been proposed for error recovery in massive MIMO systems.

The signal model for a MIMO system can be expressed as $y = Hs + v$, where $s$ is the transmitted signal from $n$ antennas, and $y$ is the received signal at $m$ receiver antennas. The matrix $H$ corresponds to $m\times n$ the flat fading channel matrix the components of which are identically and independent (iid) complex Gaussian random variables. Error recovery models exploit the sparsity of the error signal, where the error directly obtained between the transmitted signal and MMSE estimated signal (hard decision) which is notated as $s_{MMSE}^d$. Then the final signal model can be obtained as \cite{7840470}
\begin{align}
    y - Hs_{MMSE}^d &= H(s - s_{MMSE}^d) + \tilde v\\
    \label{equ:MIMO}y' &= He + \tilde v.
\end{align}
where $e$ is the sparse error vector. By solving (\ref{equ:MIMO}) with the proposed CPS algorithm, we obtain an estimated version of the error $e_{MMSE}^d$, which we then use to recover the transmitted signal via \cite{7840470}
\begin{align}
    s_{Cauchy} = s_{MMSE}^d + e_{MMSE}^d.
\end{align}

For this simulation experiment, we studied the error recovery performance of the proposed method in MIMO systems of 16$\times$16 and 50$\times$50 (an example massive MIMO) for the QPSK and 16QAM modulation schemes, under various noise conditions. We compared our proposed method with the MMSE and ZF detectors and investigated the effect of the convergence condition derived in this paper. The transmitted signal has $10^5$ symbols and we run Monte Carlo simulations of size 100. The performance analysis is presented in Fig. \ref{fig:mimo2}.

When evaluating the scatterplots in Fig. \ref{fig:mimo2} (a) and (d), we can clearly see that the signal estimated via the proposed method corresponds to a good error recovery when compared to the MMSE results presented in \ref{fig:mimo2}-(a). Furthermore, the BER vs. SNR curves presented in the second and third columns show a significant SNR gain compared to the state-of-the-art. Depending on the MIMO parameters and the modulation scheme, the proposed method achieves an SNR gain of around 4 to 9 dBs at a BER value of $10^{-4}$ when compared to the ZF and MMSE estimators. Moreover, we can clearly see that increasing the shape parameter $\gamma$ improves the estimation performance of the proposed method. This result correlates well with the results presented in Figures \ref{fig:gammaComp1D} and \ref{fig:gammaComp} for two previous simulation cases. The best performance is also obtained between $\gamma$ values of $30\sqrt{\mu}/2$ and $50\sqrt{\mu}/2$.


\begin{figure*}[htbp]
\centering
\subfigure[]{\includegraphics[width=.25\linewidth]{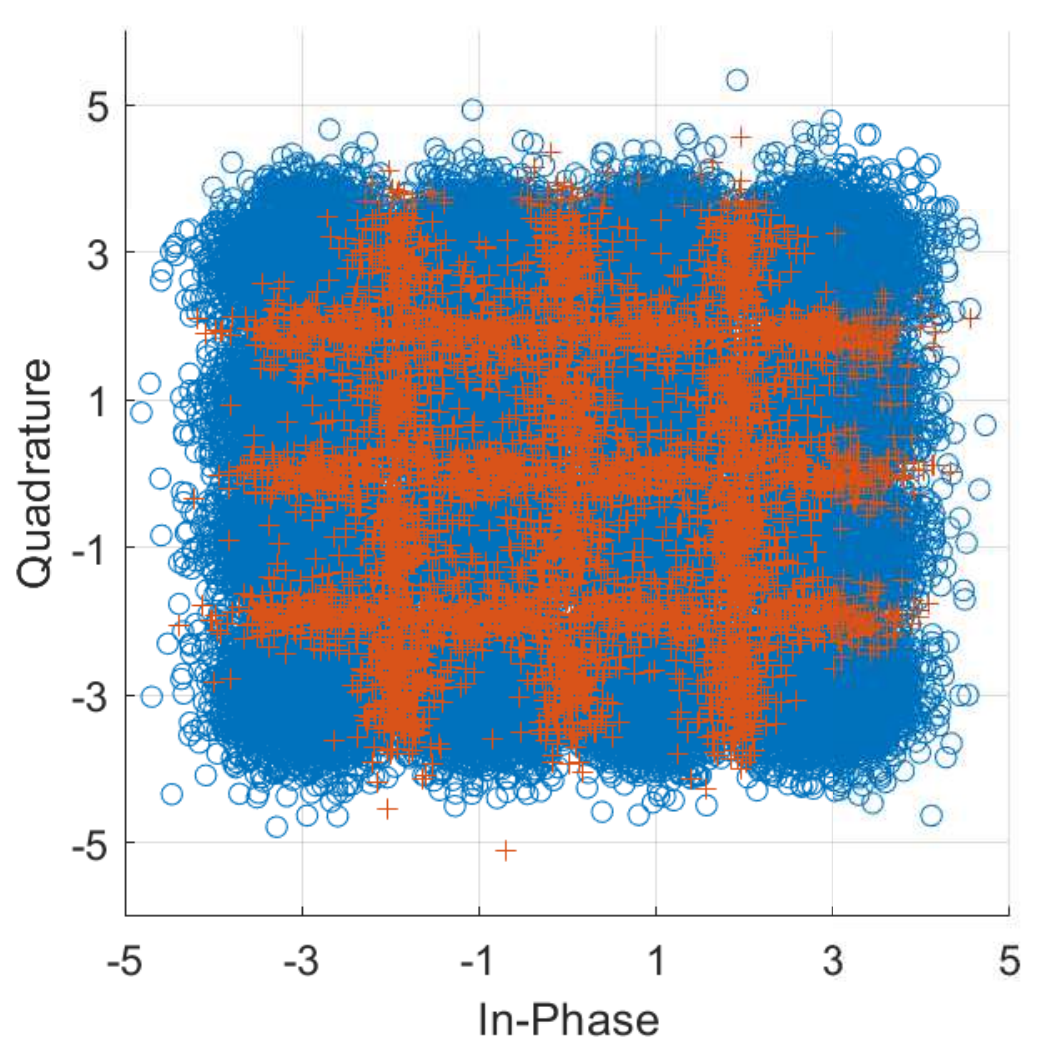}}
\subfigure[]{\includegraphics[width=.36\linewidth]{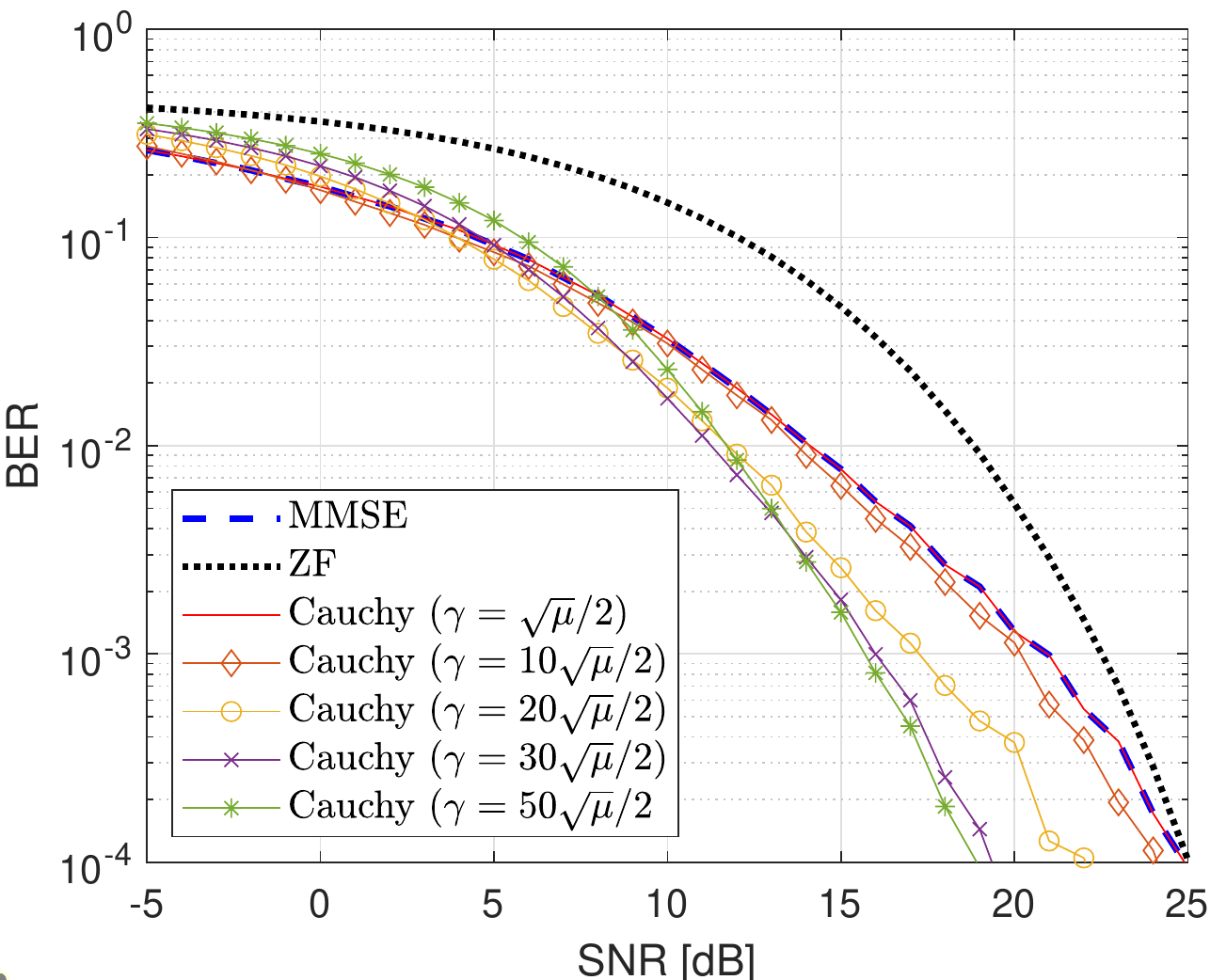}}
\subfigure[]{\includegraphics[width=.36\linewidth]{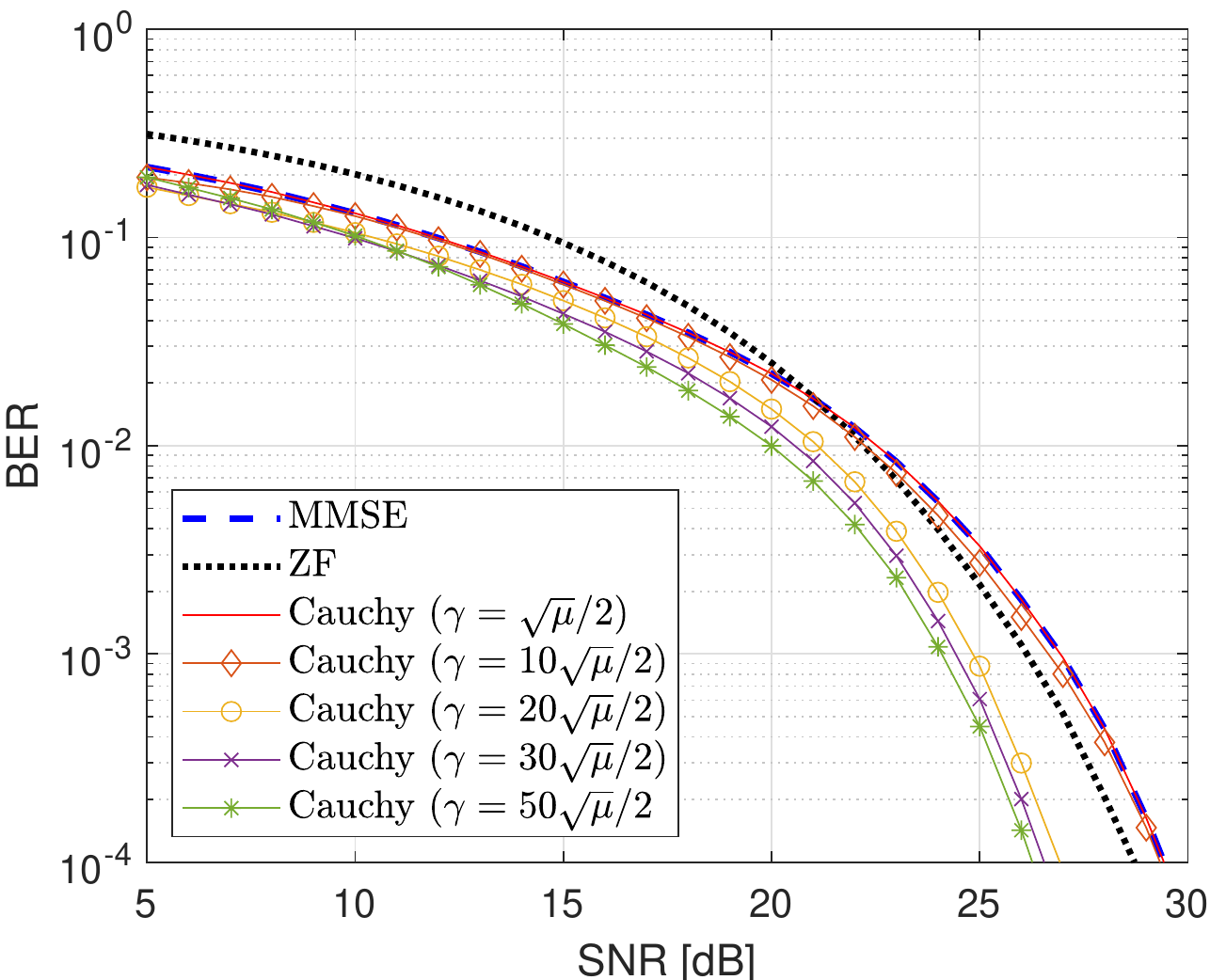}}
\subfigure[]{\includegraphics[width=.25\linewidth]{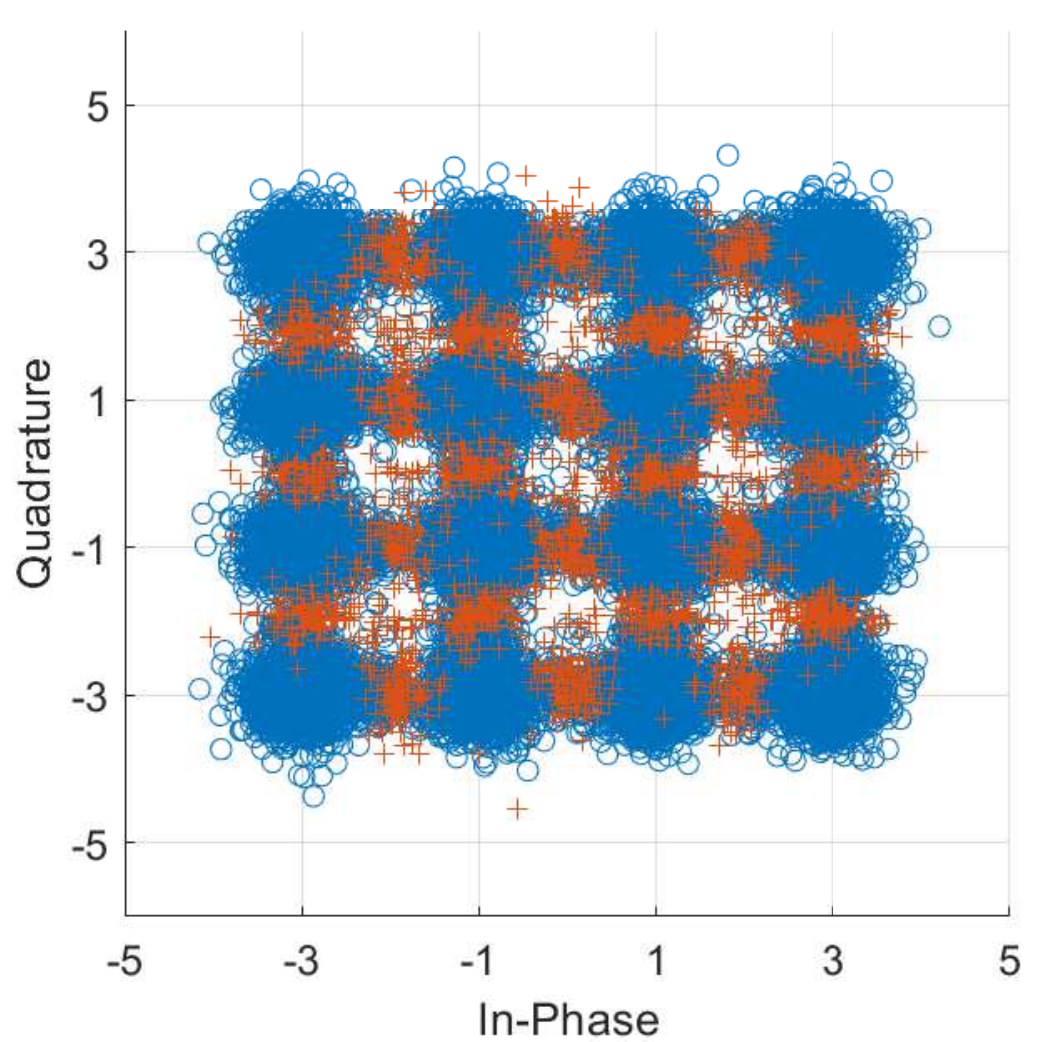}}
\subfigure[]{\includegraphics[width=.36\linewidth]{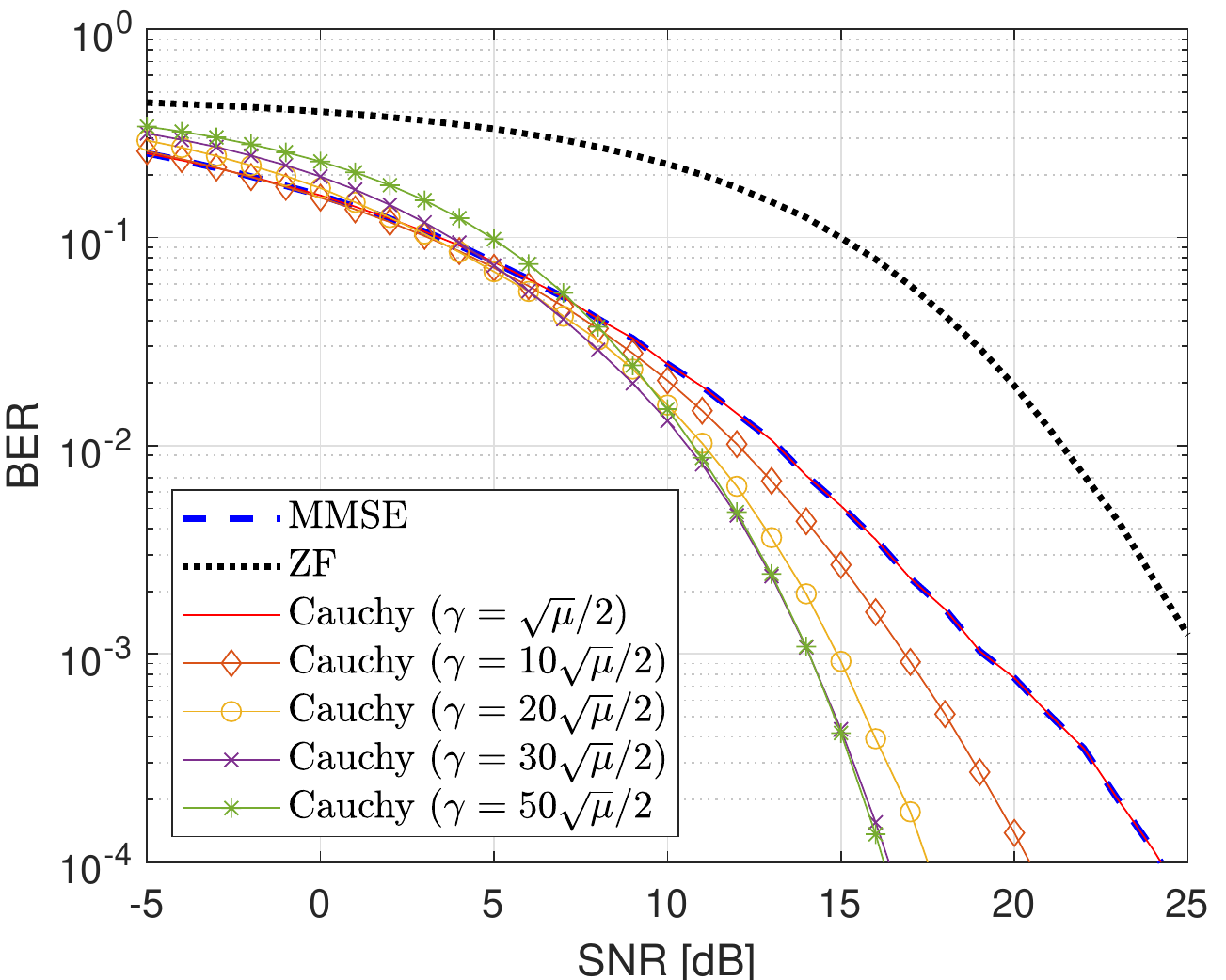}}
\subfigure[]{\includegraphics[width=.36\linewidth]{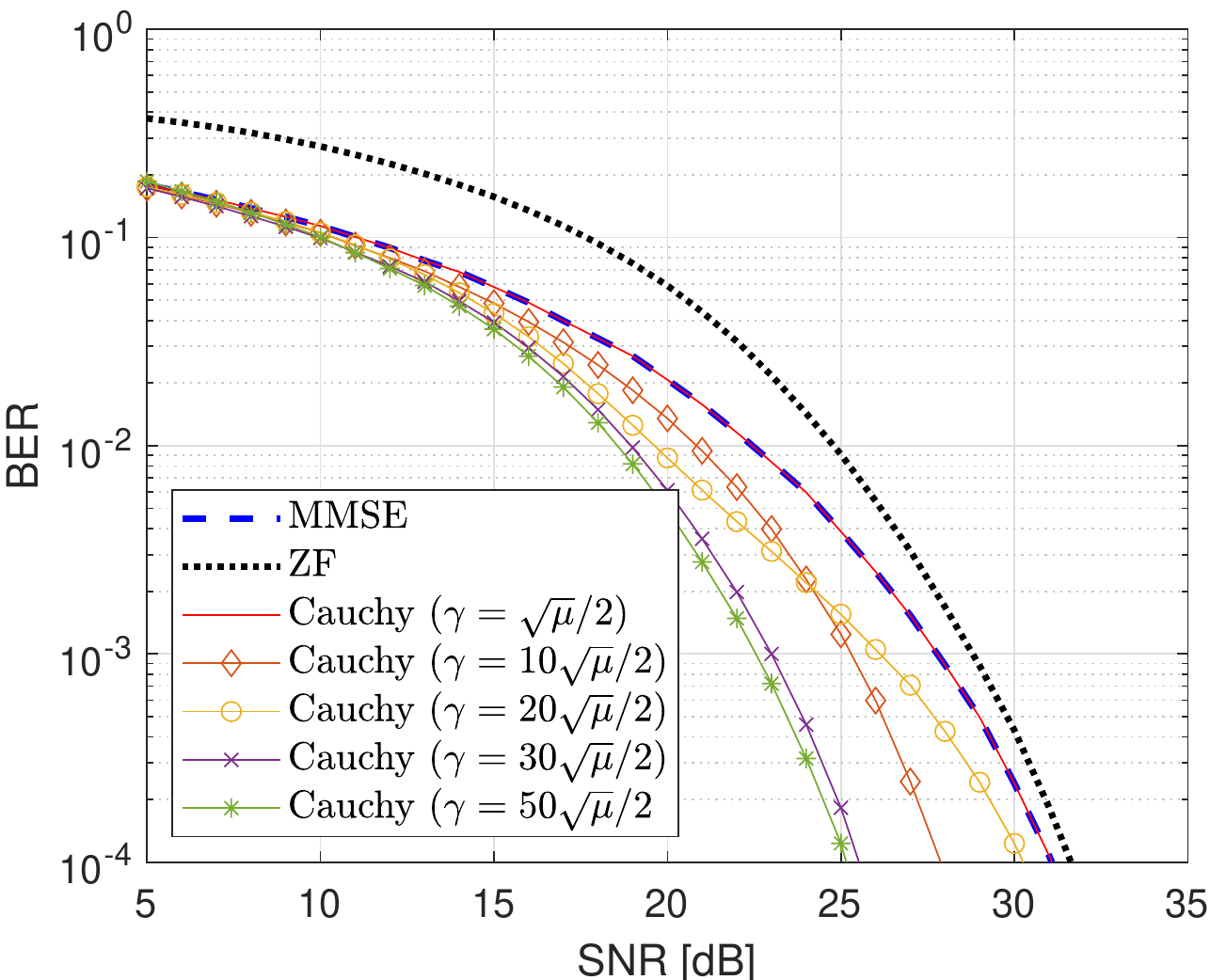}}
\caption{MIMO error recovery performance results. First column refers to scatterplots for a 50$\times$50 MIMO system with 16QAM modulation and 20dB SNR. (a) MMSE. (d) Cauchy ($\gamma = 30\sqrt{\mu}/2$). Red "+" symbols refer to erroneous detections whilst blue "o" symbols are correct detections. Second and third columns are BER vs. SNR graphs for (b) 16$\times$16 - QPSK. (c) 16$\times$16 - 16QAM. (e) 50$\times$50 - QPSK. (f) 50$\times$50 - 16QAM.}
\label{fig:mimo2}
\end{figure*}


\section{Conclusions}\label{sec:conc}
In this paper, we investigated a non-convex penalty function based on the Cauchy distribution. We proposed a FB proximal splitting methodology that employs the Cauchy proximal operator, namely the CPS algorithm. Furthermore, we derived a closed form expression for the Cauchy proximal operator. In order to guarantee the convergence of the proposed proximal splitting algorithm in spite of the non-convexity of the penalty term, we derived a condition relating the Cauchy scale parameter $\gamma$ and the step size parameter $\mu$ of the FB algorithm. Moreover, in special cases where the forward operator is orthogonal ($\mathcal{A}^T\mathcal{A} = \mathcal{I}$), or an overcomplete tight frame ($\mathcal{A}^T\mathcal{A} = r\mathcal{I}$) with $r>0$, we derived another condition for convexity that is independent on the proximal splitting algorithm employed.

In order to demonstrate the effectiveness of the proposed penalty function, we tested its performance in denoising, de-convolution and error recovery examples in comparison to the state-of-the-art methods, and $L_1$ and $TV$ norm penalty functions. The Cauchy based penalty achieved better reconstruction results compared to both penalty functions. We further showed the effect of violating the proposed condition in both examples. We concluded that the best parameter set always lays in the correct side of the derived critical value (i.e. $\gamma \geq \frac{\sqrt{\mu}}{2}$).

Our current work is focused on investigating the use of the Cauchy proximal operator in representation learning via convolutional sparse coding and will be reported in a future communication.
In addition, the existence of a closed-form expression for the Cauchy proximal operator makes is suitable for advanced Bayesian inferences, such as uncertainty quantification, e.g. via $p$-MCMC methods, which is another of our current endeavours. Finally, since other statistical models could lead to equally valid penalty functions and corresponding proximal operators, in our current work we are also investigating such alternatives, including the hyper-Laplacian, Pareto, and the Student's t-distribution.

\bibliographystyle{IEEEtran}
\bibliography{Cauchy_SR}

\end{document}